\documentclass[journal]{IEEEtran}

\usepackage{amsmath,amssymb,amsthm,mathtools}
\usepackage{bm}
\usepackage{bbm}
\usepackage{mathrsfs}
\usepackage{latexsym}
\usepackage{graphicx}
\usepackage{cite}
\usepackage{psfrag}
\usepackage{array}
\usepackage{makecell}
\usepackage{siunitx}
\usepackage{color,comment}
\usepackage{float}
\usepackage{stfloats}

\usepackage{algorithm}
\usepackage{algorithmic}

\ifCLASSOPTIONcompsoc
    \usepackage[caption=false,font=normalsize,labelfont=sf,textfont=sf,labelformat=simple]{subfig}
\else
    \usepackage[caption=false,font=footnotesize,labelformat=simple]{subfig}
\fi

\newtheorem{theorem}{Theorem}

\theoremstyle{definition}
\newtheorem{definition}{Definition}

\theoremstyle{remark}

\allowdisplaybreaks

\begin{document}
\bstctlcite{IEEEexample:BSTcontrol}

\title{ISAC with Co-Prime Arrays: Virtual-Aperture Sensing and Uplink--Downlink Communications}

\author{{Jing Zhang,~\IEEEmembership{Member,~IEEE,}
Yuxiao Liu,
Jiayi Sun,\\
Junliang Ye,~\IEEEmembership{Member,~IEEE,}
and Derrick Wing Kwan Ng,~\IEEEmembership{Fellow,~IEEE}%
\thanks{Part of this work has been presented at the IEEE PIMRC 2026.}}}

\IEEEaftertitletext{\vspace{-0.8cm}}

\maketitle

\pagestyle{empty}
\thispagestyle{empty}

\begin{abstract}
Integrated sensing and communication (ISAC) enables simultaneous communication and environmental sensing in unmanned aerial vehicle (UAV) networks, but its performance is constrained by the physical antenna aperture and residual self-interference (SI) in full-duplex (FD) sensing. To address these issues, we propose a shared-aperture ISAC architecture in which a sparse co-prime array (CPA) is embedded in a uniform linear array (ULA) grid for FD sensing, while the remaining antenna positions support time-division duplexing (TDD) communication. We characterize the sensing performance through an order-wise Cramer--Rao bound (CRB) analysis, showing that the CPA achieves a stronger asymptotic sensing gain than the partitioned ULA benchmark in both single-target and nondegenerate multi-target scenarios. We further reveal a space--time sampling tradeoff under the same physical aperture. Based on the proposed architecture, we formulate a non-convex joint resource allocation problem that maximizes the weighted downlink--uplink sum rate by jointly designing the sensing transmit covariance, downlink precoder, and uplink receive beamformers under sensing accuracy, BS transmit-power, communication QoS, and residual SI constraints. An alternating-optimization-based algorithm is developed. Simulations demonstrate consistent performance gains over the considered baselines and confirm the complementary benefits of the CPA virtual aperture and sensing covariance optimization.
\end{abstract}

	\IEEEpeerreviewmaketitle

\vspace{-0.5cm}
\section{introduction}

Integrated sensing and communication (ISAC) has emerged as a key enabling paradigm for unlocking the full potential of future sixth-generation (6G) wireless networks, where communication systems are expected to provide not only high-rate data transmission but also native sensing and environmental awareness capabilities \cite{9390169}. On the one hand, ISAC enables effective spectrum and hardware reuse, allowing radar sensing and multiuser communication to be jointly supported within a unified multi-input multi-output (MIMO) transmit platform \cite{10680586}. On the other hand, it endows wireless networks with radio perception capabilities, thereby providing environmental information that can facilitate beam management, mobility tracking, and resource allocation in dynamic scenarios \cite{9296833}. Owing to these advantages, ISAC is considered promising for a broad range of emerging applications, including autonomous driving, intelligent transportation, industrial Internet of Things (IIoT), smart robotics, and low-altitude aerial networks \cite{9390169,10012421,10418473}. 

In general, since sensing and communication signals are engineered for different objectives, ISAC systems usually require careful resource allocation and signal processing across the time, frequency, spatial, and waveform domains. In fact, existing studies have evaluated the performance gains obtained by sharing spectrum resources and radio-frequency (RF) hardware at the physical layer of ISAC systems to achieve goals such as improving spectral efficiency, reducing costs, decreasing hardware size and weight, and enabling mutual assistance and enhancement between communication and sensing. For spectrum-efficient integration, dual-function radar-communication (DFRC)
systems have been investigated as a promising solution to the spectrum
congestion problem \cite{8828023}. From the hardware
integration perspective, DFRC designs have shown that implementing sensing
and communication on the same hardware platform can substantially reduce the duplicated
hardware costs compared with separately deployed radar and communication
systems \cite{9093221}. Moreover, ISAC has been recognized as
a viable approach to reducing the number of antennas, system size, weight,
and power consumption, by reusing RF hardware and physical-layer resources
\cite{9737357}. In addition to such integration gain, ISAC can also
enable coordination gains, where sensing information is exploited to assist
communication tasks such as predictive beamforming and beam tracking
\cite{9171304}.

Although hardware-level integration of ISAC systems can reduce overall power consumption, system size, and information exchange latency, several open issues remain unsolved, posing obstacles for achieving a seamless and effective reuse of baseband and RF transceiver hardware of ISAC systems. In general, there is a fundamental mismatch between the dynamic range requirements of communication and radar systems. Specifically, for the transmitter, the sensing function requires a high-power amplifier (PA) to transmit constant-envelope single-carrier signals such as frequency modulation continuous wave (FMCW) signals. However, communication requires a larger operating dynamic range to linearly and with low distortion transmit wide band multi-carrier signals such as orthogonal frequency division multiplexing (OFDM) \cite{9585321}. Another fundamental issue is the waveform design for integrated sensing and communication signals. In practice, OFDM waveforms in communication system usually suffer from a non-constant envelope issue \cite{6476061}, resulting in a high peak-to-average power ratio (PAPR). Meanwhile, sensing systems typically employ power amplifiers (PAs) with a small dynamic range, which are nonlinear amplifiers generally operating near saturation to achieve high transmission efficiency. When an integrated OFDM waveform signal passes through such a PA, it inevitably undergoes a certain degree of nonlinear distortion, leading to reduced transmitter power efficiency, distortion of the OFDM signal, and degradation of both radar and communication performance \cite{7944478,7944480}.
% Meanwhile, to reduce the implementation complexity associated with receiver-side separation of sensing and communication signals, time-division and frequency-division schemes are usually adopted. In \cite{9858656}, a UAV-enabled integrated periodic sensing and communication mechanism was proposed, where sensing tasks are periodically performed over selected time slots while UAV trajectory, sensing-time selection, user association, and beamforming are jointly optimized.

To address these issues, ISAC systems typically employ time-division duplexing (TDD) schemes, referred to as TDD-ISAC, where the receiver and transmitter operate in a time-division manner. For continuous-wave sensing systems, the receiver and transmitter are required not only to operate almost simultaneously, but also continuously in order to constantly detect echo signals. Meanwhile, communication capacity is also potentially degraded due to dedicated time slots reserved for sensing echoes. Indeed, the temporal separation between communication and sensing prevents simultaneous operation, thereby constraining real-time responses to highly dynamic environmental changes. To enable communication and sensing to fully reuse the transceiver chain, full-duplex (FD) technology combined with self-interference cancellation allows ISAC devices to transmit and receive simultaneously, effectively eliminating the sensing coverage gaps and improving the communication capacity \cite{9724187,10663787}. 
To effectively cancel self-interference (SI), the transmitting and receiving antennas of sensing system usually separates and joint transmitter-receiver beamforming is employed to align the null of the transmit antenna pattern with the main lobe of the receive antenna. The antenna switch enables dynamic configuration of transmit and receive functionalities within the same antenna array, is discussed in some works.
In \cite{10810291}, a dynamically partitioned monostatic multiple-input multiple-output (MIMO)-ISAC array was proposed, where each base station (BS) antenna element can be flexibly configured as either a transmit or receive antenna to minimize direction of arrival (DOA) estimation error while satisfying communication signal to interference plus noise ratio (SINR) and power constraints.  Also, in \cite{gao2026llmenabledantennapartitioningbeamforming}, a segmented pinching-antenna-assisted ISAC architecture was proposed, where segment-wise transmit/receive partitioning, antenna deployment, and beamforming are jointly optimized for varying user and target configurations.

Although separating communication and sensing functions can simplify hardware architecture and receiver-side signal processing, several fundamental obstacles remain to achieving high data-rate communication in practical ISAC systems. On the one hand, the limited BS transmit power must be judiciously allocated between downlink communication and dedicated sensing waveforms. Allocating more power to sensing can strengthen the echo signals and improve estimation accuracy, but it leaves less transmit power for downlink communication and may also increase the residual SI experienced during uplink reception. In contrast, allocating excessive BS power to downlink communication reduces the resources available for target illumination, thereby requiring a careful sensing--communication power tradeoff. On the other hand, in multi-target and multiuser scenarios, sensing and communication signals coexist within the same spatial aperture and introduce coupled interference effects. Therefore, an effective ISAC architecture should jointly consider power allocation, interference suppression, and spatial degrees of freedom to simultaneously support high-resolution sensing and reliable high-rate communication.

In this paper, we investigate a co-located TDD communication and FD sensing ISAC system for UAV networks, where sensing and communication share a common physical aperture. Specifically, the BS employs an active co-prime array (CPA) sensing structure embedded in a common uniform linear array (ULA) grid. The sensing transmit and receive subarrays exploit the enlarged virtual aperture induced by the CPA geometry, while the remaining antenna positions are reused as TDD communication transceiver antennas for downlink transmission and uplink reception. This shared-aperture architecture enables continuous FD sensing while supporting bidirectional communication without requiring an additional dedicated communication array.

We analytically characterize the sensing performance of the proposed architecture by deriving order-wise bounds on the angle-related CRB and comparing its asymptotic behavior with that of a ULA benchmark. We further quantify the associated space--time sampling tradeoff in terms of the required number of sensing snapshots. Building upon the proposed architecture, we formulate a joint resource allocation problem to maximize the weighted downlink--uplink communication sum rate by jointly designing the sensing transmit precoder, downlink communication precoder, and uplink receive beamformers, subject to the sensing accuracy, total BS transmit-power, communication QoS, and residual SI constraints. To tackle the resulting non-convex problem, we develop an AO-based iterative algorithm incorporating fractional programming, closed-form uplink receive beamforming, convex downlink precoder optimization, and semidefinite programming for the sensing transmit covariance.

The main contributions of this paper are summarized as follows.

\begin{itemize}
\item We propose a shared-aperture ISAC antenna architecture, where a sparse co-prime sensing array is embedded into a common ULA grid and the remaining antenna positions within the same physical aperture are reused as TDD communication transceiver antennas. This architecture enables continuous high-resolution FD sensing while supporting both downlink transmission and uplink reception in UAV networks.

\item We establish an order-wise CRB analysis for the proposed CPA-based sensing architecture and conduct an asymptotic comparison with a ULA benchmark in both single-target and nondegenerate multi-target scenarios. The results demonstrate that the CPA geometry provides a stronger asymptotic sensing advantage, thereby offering theoretical support for the CRB improvement enabled by the enlarged virtual aperture.

\item We further identify and quantify a space--time sampling tradeoff under a fixed physical aperture. Although the sparse CPA exploits an enlarged virtual aperture with fewer physical sensing elements, its estimation-accuracy loss relative to the same-aperture ULA generally needs to be compensated for increasing the number of snapshots.

\item We formulate a joint resource allocation problem for the proposed shared-aperture architecture to maximize the weighted downlink--uplink communication sum rate by jointly optimizing the sensing transmit design, downlink precoder, and uplink receive beamformers under the sensing CRB, total BS transmit-power, downlink/uplink QoS, and residual SI constraints. An AO-based iterative algorithm is developed by combining fractional programming, closed-form uplink receive beamforming, convex downlink precoder optimization, and semidefinite programming for the sensing transmit covariance.
\end{itemize}

% The remainder of this paper is organized as follows. Section~II
% introduces the system model and describes the considered ISAC architecture. In
% Section~III, we formulate a joint sensing-communication beamforming optimization problem
% that minimizes the sensing CRB, subject to downlink communication SINR constraints and a total transmit power budget. Section~IV develops a tractable convex SDP reformulation that efficiently solve the proposed non-convex optimization problem. Simulation results and performance comparisons with baselines are
% provided in Section~V, and finally, Section VI concludes the paper.

The remainder of this paper is organized as follows. Section~II introduces the system model and the proposed shared-aperture ISAC architecture. Section~III analyzes the order-wise sensing CRB and formulates the joint resource allocation problem for maximizing the weighted downlink--uplink communication sum rate under the sensing accuracy, total BS transmit-power, communication QoS, and residual SI constraints. Section~IV develops the FP-based AO solution framework for jointly optimizing the communication beamformers and sensing transmit design. Simulation results and performance comparisons with the baseline schemes are presented in Section~V, and Section~VI concludes the paper.

\textit{Notation:} $(\cdot)^{\mathsf T}$, $(\cdot)^{\mathsf H}$, and $(\cdot)^*$ denote the transpose, Hermitian transpose, and element-wise complex conjugate, respectively. 
$\operatorname{tr}(\cdot)$, $\operatorname{diag}(\cdot)$, $\operatorname{vec}(\cdot)$, and $\operatorname{rank}(\cdot)$ denote the trace, diagonalization, vectorization, and rank operators, respectively. 
$\otimes$ denotes the Kronecker product. 
$\mathbb{E}\{\cdot\}$ denotes expectation; $\mathrm{Re}\{\cdot\}$ and $\mathrm{Im}\{\cdot\}$ denote the real and imaginary parts. 
$\mathbb{C}$ and $\mathbb{R}$ denote the complex and real fields, respectively.  
$\mathcal{CN}(\mu,\sigma^2)$ denotes a circularly symmetric complex Gaussian distribution with mean $\mu$ and variance $\sigma^2$. 
$\mathbf{I}_N$ denotes the $N\times N$ identity matrix. 
$\mathbf{e}_i$ denotes the $i$-th canonical basis vector with a compatible dimension.
$\frac{\partial f}{\partial x_i}$ denotes the partial derivative with respect to $x_i$.
$\mathbf{A}\succeq\mathbf{0}$ means that $\mathbf{A}$ is Hermitian positive semidefinite.

	\section{System Model}\label{sec2}

    \begin{figure}[t]
    \centering
    \includegraphics[width=\linewidth]{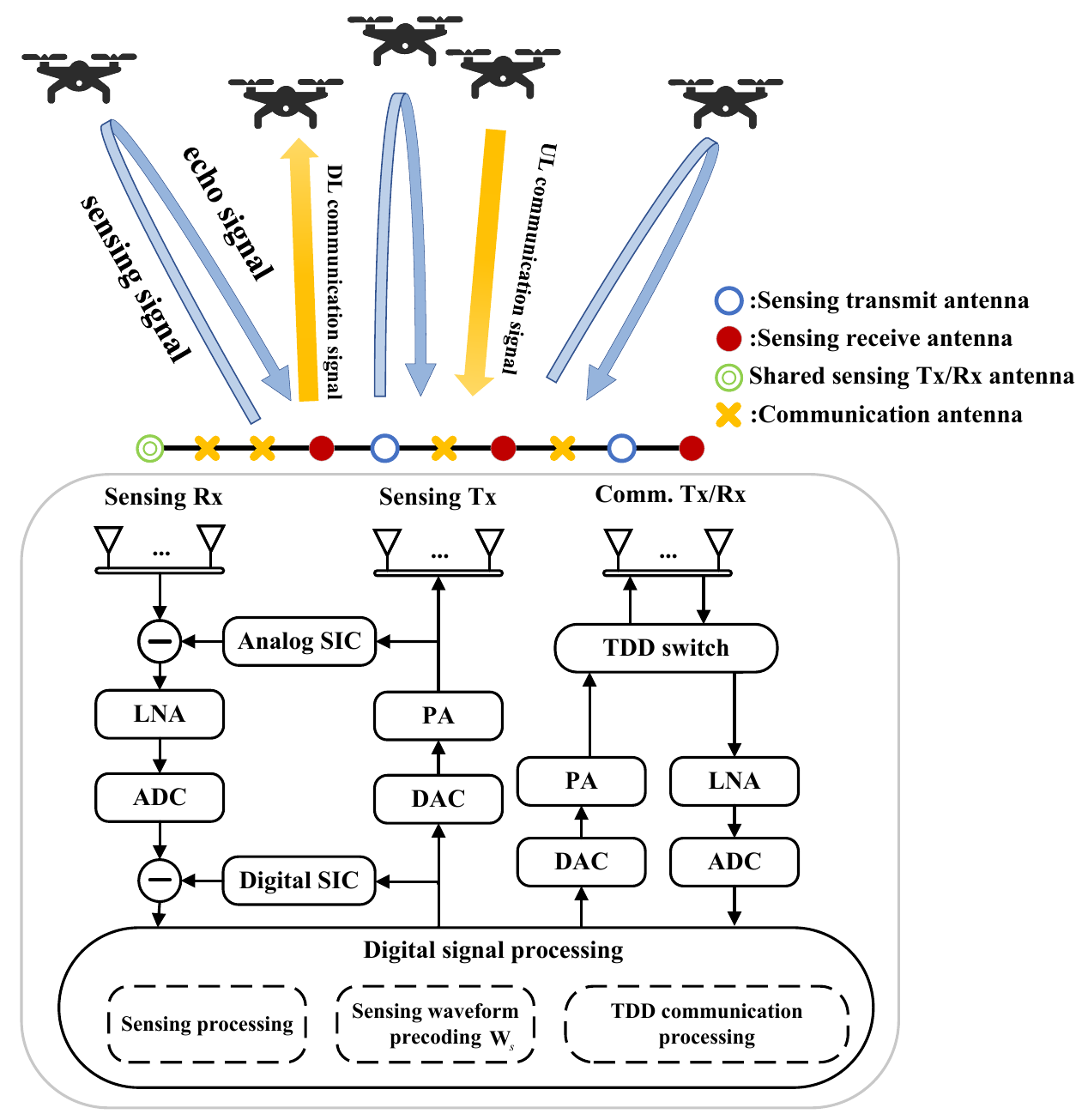}%
    \caption{System model of the considered shared-aperture ISAC UAV network with TDD uplink--downlink communication and FD sensing, where a $(M_1,M_2)=(3,4)$ co-prime sensing array is embedded in a 10-element ULA grid and the remaining four antenna positions are reused as communication transceiver antennas.}
    \label{fig:network model}
\end{figure}

{\color{black}
\subsection{Network Model}
As shown in Fig.~\ref{fig:network model}, we consider a co-located TDD communication and full-duplex (FD) sensing ISAC network
coordinated by a single ground base station (BS). The BS serves $K_c$ communication unmanned aerial vehicles (UAVs) in both downlink and uplink transmission slots and simultaneously senses $K_s$ UAV targets located in the
far field within a symmetric angular sector
\([-\theta_{\max},\theta_{\max}]\). Here, \(\theta_{\max}>0\) denotes the
sector half-width, and the sensing-target angles satisfy
\(\theta_m\in[-\theta_{\max},\theta_{\max}]\).
Each communication UAV is equipped with a single antenna for
downlink reception and uplink transmission in the corresponding
TDD slots.

The BS is equipped with a uniform linear array (ULA) grid
comprising \(M\) antenna positions with inter-element spacing
\(d=\lambda/2\). On this physical aperture, a subset of antenna
positions is configured as an active \((M_1,M_2)\) co-prime array
(CPA) for FD sensing, where \(M_1\) antenna elements are used for
sensing transmission and \(M_2\) antenna elements are used for sensing
reception. The remaining antenna positions are reused as TDD communication transceiver antennas. 
% The communication system implement frequency-division duplex modes, in which the uplink communication and downlink communication occupy the different spectrum. 
% In the frequency domain, let \(B\) denote the total available
% bandwidth. The FD sensing function occupies the whole bandwidth
% \(B\), while the uplink communication function is allocated
% \(B_{\rm c}=\alpha B\), where \(0<\alpha\le 1\). Accordingly, only
% the spectrum-overlapped portion of the uplink communication signal
% contributes to the leakage observed by the sensing receiver.
In this paper, sensing is performed in an active BS-sensing manner using dedicated BS-transmitted sensing waveforms. The uplink and downlink communication signals are used only for data transmission, and their leakage components are modeled as interference at the sensing receiver.
% To specify the sensing-antenna layout, \((M_1,M_2)\) denotes a pair
% of co-prime positive integers. Following the classical co-prime array
% construction, two sparse uniform linear subarrays with \(M_1\) and
% \(M_2\) antenna elements, respectively, are co-located on the same
% physical aperture, with inter-element spacings \(M_2d\) and \(M_1d\),
% respectively. The zeroth elements are aligned, and the associated
% difference coarray provides about \(M_1M_2\) distinct spatial lags.
% The number of distinct physical sensing antenna positions is given by
% $M_s=M_1+M_2-1.$
% We define
% $\rho=\frac{M_1}{M_s+1}$
% as the ratio of sensing transmit antenna elements to the total sensing
% Tx/Rx antenna elements.

\subsection{Co-prime Sensing Array}

To make the CPA deployment explicit, we first describe the
physical sensing-antenna positions. Let \((M_1,M_2)\) be a
pair of co-prime positive integers \footnote{Two positive integers are co-prime if their only common
factor is 1. For example, 4 and 3 have factors
$\{1,2,4\}$ and $\{1,3\}$, respectively, and are therefore
co-prime. In contrast, 2 has factors $\{1,2\}$, so 4 and 2 are
not co-prime because they share the factor 2.}. Following the classical co-prime array
construction~\cite{5609222}, two sparse
ULA subarrays are placed on the same underlying ULA grid
with unit spacing \(d\). In the considered active sensing
architecture, the \(M_1\)-element subarray is used for sensing
transmission, while the \(M_2\)-element subarray is used for
sensing reception. Specifically, their physical positions are
given by
$\mathcal P_{\rm t}
=
\{mM_2d~|~m=0,\ldots,M_1-1\},$
and
$\mathcal P_{\rm r}
=
\{nM_1d~|~n=0,\ldots,M_2-1\},$
respectively.

The two subarrays are aligned at their zeroth elements, which
are shared as a common reference sensing element. Due to the
co-primality of \(M_1\) and \(M_2\), no other physical positions
overlap. Accordingly, the number of distinct physical sensing
antenna positions is
\begin{equation}
M_s=|\mathcal P_{\rm t}\cup\mathcal P_{\rm r}|
=M_1+M_2-1.
\end{equation}
The remaining antenna positions are reused as TDD communication transceiver antennas, whose number is given by $M_c=M-M_s.$

For the active Tx/Rx sensing structure considered in this
paper, the sensing transmit and receive subarrays jointly form
a virtual sensing array through their cross-sum spatial
positions. Specifically, each transmit element at
\(p_{\rm t}\in\mathcal P_{\rm t}\) and each receive element at
\(p_{\rm r}\in\mathcal P_{\rm r}\) generate an equivalent virtual
sensing position \(p_{\rm t}+p_{\rm r}\). Using the physical
positions of the two subarrays, the CPA-induced virtual
sensing array is given by
\begin{equation}
\small
\mathcal V_{\rm CPA}
=
\{(mM_2+nM_1)d~|~
m=0,\ldots,M_1-1,\;
n=0,\ldots,M_2-1\}.
\label{eq:cpa_virtual_array}
\end{equation}
Due to the co-primality of \(M_1\) and \(M_2\), these
\(M_1M_2\) cross-sum positions are distinct. Hence, with only
\(M_s=M_1+M_2-1\) distinct physical sensing antenna
positions, the proposed CPA sensing structure synthesizes
\(M_1M_2\) virtual sensing positions. This enlarged virtual
aperture improves the angular resolution without requiring a
fully populated sensing array.
Finally, the sensing Tx/Rx split ratio is defined as
$\rho=\frac{M_1}{M_s+1}.$

% \subsection{Co-prime sensing array}

% To specify the sensing-antenna layout, \((M_1,M_2)\) denotes
% a pair of co-prime positive integers. Following the classical
% co-prime array construction, the sensing transmit and receive
% subarrays are arranged as two sparse ULAs on the same underlying
% ULA grid. Specifically, the \(M_1\)-element sensing transmit
% subarray has inter-element spacing \(M_2d\), while the
% \(M_2\)-element sensing receive subarray has inter-element
% spacing \(M_1d\). Their zeroth elements are aligned as the
% common reference element. Since \(M_1\) and \(M_2\) are
% co-prime, the two sparse subarrays do not overlap except at
% this reference position. Hence, the number of distinct physical
% sensing antenna positions is given by
% $M_s=M_1+M_2-1.$

% For the active Tx/Rx sensing structure considered in this paper,
% the transmit and receive subarrays jointly induce a virtual sensing
% array through their cross-sum spatial positions. Following the
% co-prime array principle, this virtual array provides \(M_1M_2\)
% distinct virtual degrees of freedom with only \(M_s\) physical
% sensing antenna positions, thereby enlarging the effective aperture
% and improving the angular resolution. We define
% $\rho=\frac{M_1}{M_s+1}$
% as the ratio of sensing transmit antenna elements to the total
% number of sensing Tx/Rx antenna elements.

}

{\color{black}
\subsection{Signal Model}

We consider an ISAC transmission frame consisting of $L$ discrete-time symbol intervals, indexed by $n\in\{1,\ldots,L\}$. During each frame, the BS continuously performs FD sensing with the sensing transmit and receive subarrays. Meanwhile, the communication transceiver antennas operate in a TDD manner, transmitting downlink signals in the downlink slot and receiving uplink signals in the uplink slot. The sensing and communication channels are assumed to remain quasi-static within one frame due to the relatively low mobility of the UAVs over the frame duration.

% Let \(B\) denote the total available bandwidth. The FD sensing function
% occupies the whole bandwidth \(B\), while the uplink communication
% function is allocated \(B_{\rm c}=\alpha B\), where \(0<\alpha\leq 1\).
% Only the spectrum-overlapped part of the uplink communication signal
% contributes to the leakage observed by the sensing receiver.

For sensing transmission, the BS transmits dedicated sensing
waveforms through the \(M_1\)-element sensing transmit subarray.
The sensing transmit vector is given by
\begin{equation}
\mathbf x_{\rm s}(n)
=
\mathbf W_{\rm s}\mathbf s_{\rm s}(n),
\label{eq:sensing_transmit_signal}
\end{equation}
where \(\mathbf x_{\rm s}(n)\in\mathbb C^{M_1\times 1}\),
\(\mathbf W_{\rm s}\in\mathbb C^{M_1\times d_{\rm s}}\), and
\(\mathbf s_{\rm s}(n)\in\mathbb C^{d_{\rm s}\times 1}\) denote
the sensing transmit vector, sensing precoder, and sensing symbol
vector, respectively. With
\(\mathbb E\{\mathbf s_{\rm s}(n)\mathbf s_{\rm s}^{H}(n)\}
=\mathbf I_{d_{\rm s}}\), the sensing transmit covariance induced
by \(\mathbf W_{\rm s}\) is given by
\begin{equation}
\mathbf R_{\rm s}(\mathbf W_{\rm s})
=
\mathbb E\{\mathbf x_{\rm s}(n)\mathbf x_{\rm s}^{H}(n)\}
=
\mathbf W_{\rm s}\mathbf W_{\rm s}^{H}
\in\mathbb C^{M_1\times M_1}.
\label{eq:sensing_covariance}
\end{equation}

Let \(\mathbf X_{\rm s}=[\mathbf x_{\rm s}(1),\ldots,\mathbf x_{\rm s}(L)]\in
\mathbb C^{M_1\times L}\) denote the sensing waveform matrix over
\(L\) snapshots. For sufficiently large \(L\), the sample covariance
of the sensing waveform can be approximated as
\begin{equation}
\frac{1}{L}
\mathbf X_{\rm s}\mathbf X_{\rm s}^{H}
\approx
\mathbf R_{\rm s}(\mathbf W_{\rm s}).
\label{eq:sample_covariance}
\end{equation}

For downlink communication, let
$\mathbf d^{\rm DL}(n)
=
[d_1^{\rm DL}(n),\ldots,d_{K_c}^{\rm DL}(n)]^T
\in\mathbb C^{K_c\times 1}$
denote the downlink data-symbol vector intended for the $K_c$ communication UAVs, where $\mathbb E\{\mathbf d^{\rm DL}(n)(\mathbf d^{\rm DL}(n))^H\}=\mathbf I_{K_c}$. The BS applies linear downlink transmit beamforming at the communication transceiver antennas. Let
$\mathbf W_c^{\rm DL}
=
[\mathbf w_{c,1}^{\rm DL},\ldots,\mathbf w_{c,K_c}^{\rm DL}]
\in\mathbb C^{M_c\times K_c}$
denote the downlink communication precoding matrix, where $\mathbf w_{c,k}^{\rm DL}\in\mathbb C^{M_c\times 1}$ is the transmit beamforming vector for the $k$-th communication UAV. The downlink communication transmit signal is given by
\begin{equation}
\mathbf x_c^{\rm DL}(n)
=
\mathbf W_c^{\rm DL}\mathbf d^{\rm DL}(n).
\end{equation}
% Accordingly, the downlink communication transmit covariance is
% \begin{equation}
% \mathbf R_c^{\rm DL}
% \triangleq
% \mathbb E\{\mathbf x_c^{\rm DL}(n)(\mathbf x_c^{\rm DL}(n))^H\}
% =
% \mathbf W_c^{\rm DL}(\mathbf W_c^{\rm DL})^H.
% \end{equation}

For uplink communication, let
$\mathbf d^{\rm UL}(n)
=
[d_1^{\rm UL}(n),\ldots,d_{K_c}^{\rm UL}(n)]^T
\in\mathbb C^{K_c\times 1}$
denote the uplink transmit-symbol vector from the same set of communication UAVs. The uplink transmit-power covariance is given by
$\mathbf P^{\rm UL}
=
\mathbb E\{\mathbf d^{\rm UL}(n)(\mathbf d^{\rm UL}(n))^H\}
=
\operatorname{diag}(p_1^{\rm UL},\ldots,p_{K_c}^{\rm UL}),$
where $p_k^{\rm UL}$ denotes the given uplink transmit power of the $k$-th communication UAV. To detect the uplink signals, the BS applies linear receive beamforming at the communication transceiver antennas. Let
\begin{equation}
\mathbf W_c^{\rm UL}
=
[\mathbf w_{c,1}^{\rm UL},\ldots,\mathbf w_{c,K_c}^{\rm UL}]
\in\mathbb C^{M_c\times K_c}
\end{equation}
denote the uplink receive beamforming matrix, where $\mathbf w_{c,k}^{\rm UL}\in\mathbb C^{M_c\times 1}$ is the receive beamforming vector for decoding the uplink signal of the $k$-th communication UAV.
The sensing streams, downlink communication symbols, uplink communication symbols, and receiver noises are assumed to be mutually independent over different symbol intervals.

% For uplink communication, let \(d_k(n)\in\mathbb C\) denote
% the transmit signal of the \(k\)-th communication UAV, with
% \(\mathbb E\{|d_k(n)|^2\}=p_k\), where \(p_k\) is the given
% uplink transmit power. Collecting the uplink signals as
% \(\mathbf d(n)=[d_1(n),\ldots,d_{K_{\rm c}}(n)]^T
% \in\mathbb C^{K_{\rm c}\times1}\), we have
% \(\mathbf P\triangleq
% \mathbb E\{\mathbf d(n)\mathbf d^H(n)\}
% ={\rm diag}(p_1,\ldots,p_{K_{\rm c}})\). To detect the
% uplink signals, the BS applies linear receive beamforming at the
% communication receive subarray. Let
% \(\mathbf w_{{\rm c},k}\in\mathbb C^{M_{\rm c}\times1}\)
% denote the receive beamforming vector for decoding the signal
% of the \(k\)-th communication UAV. Collecting all receive
% beamformers, we define
% \[
% \mathbf W_{\rm c}
% =
% [\mathbf w_{{\rm c},1},\ldots,\mathbf w_{{\rm c},K_{\rm c}}]
% \in\mathbb C^{M_{\rm c}\times K_{\rm c}}.
% \]
% The sensing streams, uplink communication symbols, and receiver
% noises are mutually independent over different symbol intervals.

For the sensing link, the BS receives target echoes through the
\(M_2\)-element sensing receive subarray. The UAV sensing targets
are modeled as point targets. For the \(K_s\) targets with angles
\(\{\theta_m\}_{m=1}^{K_s}\) and complex reflection coefficients
\(\{\beta_m\}_{m=1}^{K_s}\), the target echo component is determined
by the sensing transmit and receive steering vectors
\(\mathbf a_t(\theta)\in\mathbb C^{M_1\times1}\) and
\(\mathbf a_r(\theta)\in\mathbb C^{M_2\times1}\), respectively.

Due to the simultaneous operation of the sensing transmit and
receive subarrays, the direct leakage from the sensing transmit
antennas to the sensing receive antennas introduces SI. In practical
FD-ISAC receivers, SIC is usually implemented in a cascaded manner,
where analog/RF-domain SIC is first performed before the ADC and
digital-domain SIC is further applied after sampling by exploiting the
known transmitted signal~\cite{11050889,7421941}. However, due to imperfect SI
channel estimation, hardware impairments, nonlinear distortions, and
multipath leakage, the SI cannot be completely removed in practice.
Following the residual SI modeling adopted in FD systems, the
post-SIC residual SI at the sensing receiver is represented by
\(\mathbf H_{{\rm SI},s}\mathbf x_s(n)\), where
\(\mathbf H_{{\rm SI},s}\in\mathbb C^{M_2\times M_1}\) denotes the
equivalent residual SI channel from the sensing transmit subarray to
the sensing receive subarray. 
The received signal at the sensing receive subarray in symbol interval $n$ is given by
\begin{equation}
\small
\mathbf y_s(n)
=
\sum_{m=1}^{K_s}
\beta_m
\mathbf a_r^*(\theta_m)
\mathbf a_t^H(\theta_m)
\mathbf x_s(n)
+
\boldsymbol{\ell}_s(n)
+
\mathbf H_{{\rm SI},s}\mathbf x_s(n)
+
\mathbf z_s(n),
\end{equation}
where $\mathbf y_s(n)\in\mathbb C^{M_2\times 1}$ denotes the received sensing signal, and $\mathbf z_s(n)\sim\mathcal{CN}(\mathbf 0,\sigma_s^2\mathbf I_{M_2})$ is the sensing receiver noise.
The aggregate leakage term $\boldsymbol{\ell}_s(n)$ captures the effective residual interference from uplink and downlink communication signals. For CRB evaluation, we define its covariance as
$\mathbf R_{\ell,s}\triangleq\mathbb E\{\boldsymbol{\ell}_s(n)\boldsymbol{\ell}_s^H(n)\}.$

For the downlink communication link, let
$\mathbf h_{c,k}^{\rm DL}\in\mathbb C^{M_c\times 1}$
denote the channel from the communication transceiver antennas
to the $k$-th communication UAV, and let
$\mathbf h_{s,k}^{\rm DL}\in\mathbb C^{M_1\times 1}$
denote the interference channel from the sensing transmit
subarray to the $k$-th communication UAV. In the downlink
slot, the received signal at the $k$-th communication UAV is
given by
\begin{equation}
y_k^{\rm DL}(n)
=
(\mathbf h_{c,k}^{\rm DL})^H\mathbf x_c^{\rm DL}(n)
+
(\mathbf h_{s,k}^{\rm DL})^H\mathbf x_s(n)
+
z_k^{\rm DL}(n),
\end{equation}
where $z_k^{\rm DL}(n)\sim\mathcal{CN}(0,\sigma_{\rm DL}^2)$
denotes the downlink receiver noise.

For the uplink communication link, let
$\mathbf H_c^{\rm UL}
=
[\mathbf h_{c,1}^{\rm UL},\ldots,\mathbf h_{c,K_c}^{\rm UL}]
\in\mathbb C^{M_c\times K_c}$
denote the uplink channel matrix from the communication UAVs
to the communication transceiver antennas. In the uplink slot, the communication transceiver antennas receive uplink signals from the $K_c$ communication UAVs, while the concurrent FD sensing transmission may cause residual SI after SIC. Therefore,
the received uplink communication signal is given by
\begin{equation}
\mathbf y_c^{\rm UL}(n)
=
\mathbf H_c^{\rm UL}\mathbf d^{\rm UL}(n)
+
\mathbf H_{{\rm SI},c}\mathbf x_s(n)
+
\mathbf z_c^{\rm UL}(n),
\end{equation}
where $\mathbf y_c^{\rm UL}(n)\in\mathbb C^{M_c\times 1}$
denotes the received uplink communication signal, and
$\mathbf z_c^{\rm UL}(n)\sim\mathcal{CN}(\mathbf 0,\sigma_{\rm UL}^2\mathbf I_{M_c})$
is the additive noise at the communication receiver.

Based on the uplink receive beamforming model, the detected
signal of the $k$-th communication UAV is given by
$\hat d_k^{\rm UL}(n)
=
(\mathbf w_{c,k}^{\rm UL})^H\mathbf y_c^{\rm UL}(n).$

{\color{black}
\subsection{System Performance}
\noindent\textit{1) Communication Rate:}
Based on the downlink received signal in Section~II-C, the
SINR of the $k$-th communication UAV in the downlink slot is
given by
\begin{equation}
\gamma_k^{\rm DL}
=
\frac{
\left|(\mathbf h_{c,k}^{\rm DL})^H\mathbf w_{c,k}^{\rm DL}\right|^2
}{
\sum_{j\ne k}^{K_c}
\left|(\mathbf h_{c,k}^{\rm DL})^H\mathbf w_{c,j}^{\rm DL}\right|^2
+
(\mathbf h_{s,k}^{\rm DL})^H\mathbf R_s\mathbf h_{s,k}^{\rm DL}
+
\sigma_{\rm DL}^2
}.
\end{equation}
Accordingly, the downlink sum rate is expressed as
\begin{equation}
R_{\rm DL}
=
B\sum_{k=1}^{K_c}
\log_2(1+\gamma_k^{\rm DL}).
\end{equation}

Based on the uplink received signal and the receive beamforming
model in Section~II-C, the SINR of the $k$-th communication
UAV in the uplink slot is given by
\begin{equation}
\gamma_k^{\rm UL}
=
\frac{
p_k^{\rm UL}
\left|(\mathbf w_{c,k}^{\rm UL})^H\mathbf h_{c,k}^{\rm UL}\right|^2
}{
\begin{aligned}
&\sum_{j\ne k}^{K_c}
p_j^{\rm UL}
\left|(\mathbf w_{c,k}^{\rm UL})^H\mathbf h_{c,j}^{\rm UL}\right|^2 \\
&\quad+
(\mathbf w_{c,k}^{\rm UL})^H
\mathbf H_{{\rm SI},c}\mathbf R_s\mathbf H_{{\rm SI},c}^H
\mathbf w_{c,k}^{\rm UL}
+
\sigma_{\rm UL}^2\|\mathbf w_{c,k}^{\rm UL}\|^2
\end{aligned}
}.
\end{equation}
The corresponding uplink sum rate is
\begin{equation}
R_{\rm UL}
=
B\sum_{k=1}^{K_c}
\log_2(1+\gamma_k^{\rm UL}).
\end{equation}

Thus, the weighted communication sum rate is
defined as
$R_{\rm sum}
=
\alpha_{\rm DL}R_{\rm DL}
+
\alpha_{\rm UL}R_{\rm UL},$
where $\alpha_{\rm DL}\ge 0$ and $\alpha_{\rm UL}\ge 0$ are fixed
weights for balancing the downlink and uplink rates, respectively.

\noindent\textit{2) CRB-Based Sensing Performance:}
The sensing performance is characterized by the angle-related
Cramér--Rao bound (CRB), which provides a lower bound on the
estimation error covariance of the target angles. For CRB evaluation, we collect the sensing received signals over
$L$ snapshots as
$\mathbf Y_s=[\mathbf y_s(1),\ldots,\mathbf y_s(L)]
\in\mathbb C^{M_2\times L}$. Let
$\mathbf L_s=[\boldsymbol{\ell}_s(1),\ldots,\boldsymbol{\ell}_s(L)]$
and
$\mathbf Z_s=[\mathbf z_s(1),\ldots,\mathbf z_s(L)]$.
Define
$\mathbf A_t(\boldsymbol\theta)
=
[\mathbf a_t(\theta_1),\ldots,\mathbf a_t(\theta_{K_s})]$,
$\mathbf A_r(\boldsymbol\theta)
=
[\mathbf a_r(\theta_1),\ldots,\mathbf a_r(\theta_{K_s})]$,
and
$\mathbf B(\boldsymbol\beta)=\operatorname{diag}(\boldsymbol\beta)$,
where $\boldsymbol\beta=[\beta_1,\ldots,\beta_{K_s}]^T$.
Then, the stacked sensing received signal is written as
\begin{equation}
\mathbf Y_s
=
\mathbf A_r^*(\boldsymbol\theta)
\mathbf B(\boldsymbol\beta)
\mathbf A_t^H(\boldsymbol\theta)
\mathbf X_s
+
\mathbf L_s
+
\mathbf H_{{\rm SI},s}\mathbf X_s
+
\mathbf Z_s .
\end{equation}
By vectorizing $\mathbf Y_s$, we obtain
\begin{equation}
\mathbf y_s
\triangleq
\operatorname{vec}(\mathbf Y_s)
=
(\mathbf X_s^T\otimes \mathbf I_{M_2})
\sum_{m=1}^{K_s}
\beta_m\boldsymbol\phi(\theta_m)
+
\mathbf v_s,
\end{equation}
where
$\boldsymbol\phi(\theta)
=
\mathbf a_t^*(\theta)\otimes \mathbf a_r^*(\theta)
\in\mathbb C^{M_1M_2\times 1}$
denotes the CPA-induced virtual array manifold, and
$\mathbf v_s
=
\operatorname{vec}
(
\mathbf L_s
+
\mathbf H_{{\rm SI},s}\mathbf X_s
+
\mathbf Z_s
)$
collects the aggregate communication leakage, the residual SI
after SIC, and the sensing receiver noise. At each snapshot, the
corresponding non-target component is
\begin{equation}
\mathbf v_s(n)
=
\boldsymbol{\ell}_s(n)
+
\mathbf H_{{\rm SI},s}\mathbf x_s(n)
+
\mathbf z_s(n).
\end{equation}
Under the independence assumptions in the signal model, the
covariance matrix of $\mathbf v_s(n)$ is given by
\begin{equation}
\mathbf R_{v,s}
\triangleq
\mathbb E\{\mathbf v_s(n)\mathbf v_s^H(n)\}
=
\mathbf R_{\ell,s}
+
\mathbf H_{{\rm SI},s}\mathbf R_s\mathbf H_{{\rm SI},s}^H
+
\sigma_s^2\mathbf I_{M_2}.
\end{equation}

% The sensing performance is characterized by the angle-related
% Cram\'er--Rao bound (CRB), which provides a lower bound on
% the estimation error covariance of the target angles. Based on
% the sensing received signal model, the non-target components
% at the sensing receive subarray consist of the uplink communication
% leakage, the residual SI after SIC, and the receiver noise. Thus,
% the effective interference-plus-noise term can be written as
% \begin{equation}
% \mathbf v_{\rm s}(n)
% =
% \sqrt{\alpha}\mathbf H_{\rm s}\mathbf d(n)
% +
% \mathbf H_{{\rm SI},{\rm s}}\mathbf x_{\rm s}(n)
% +
% \mathbf z_{\rm s}(n).
% \label{eq:effective_noise_sensing}
% \end{equation}

% Under the independence assumptions in the signal model, the covariance
% matrix of \(\mathbf v_{\rm s}(n)\) is given by
% \begin{equation}
% \begin{aligned}
% \mathbf R_{{\rm v},{\rm s}}
% &=
% \mathbb E\{\mathbf v_{\rm s}(n)\mathbf v_{\rm s}^{H}(n)\}
% \\
% &=
% \alpha\mathbf H_{\rm s}\mathbf P\mathbf H_{\rm s}^{H}
% +
% \mathbf H_{{\rm SI},{\rm s}}
% \mathbf W_{\rm s}\mathbf W_{\rm s}^{H}
% \mathbf H_{{\rm SI},{\rm s}}^{H}
% +
% \sigma_{\rm s}^{2}\mathbf I_{M_2}.
% \end{aligned}
% \label{eq:sensing_interference_covariance}
% \end{equation}
% Here, the first term represents the covariance of the uplink
% communication leakage, the second term represents the covariance of
% the residual SI induced by the sensing precoder after SIC, and the
% third term represents the sensing receiver noise covariance.

Let the unknown target parameter vector be defined as
$\boldsymbol\xi
=
[
\boldsymbol\theta^T,
\operatorname{Re}\{\boldsymbol\beta\}^T,
\operatorname{Im}\{\boldsymbol\beta\}^T
]^T
\in\mathbb R^{3K_s\times 1},$
where \(\boldsymbol\theta=[\theta_1,\ldots,\theta_{K_s}]^T\) denotes
the target-angle vector and
\(\boldsymbol\beta=[\beta_1,\ldots,\beta_{K_s}]^T\) denotes the
complex reflection-coefficient vector. The mean vector of the received
sensing signal is expressed as
$\boldsymbol\mu(\boldsymbol\xi)
=
\operatorname{vec}
\left(
\mathbf A_{\rm r}^{*}(\boldsymbol\theta)
\mathbf B(\boldsymbol\beta)
\mathbf A_{\rm t}^{H}(\boldsymbol\theta)
\mathbf X_{\rm s}
\right).$

Due to the effective interference-plus-noise covariance is generally
colored, the FIM of \(\boldsymbol\xi\) is given by
\begin{equation}
[\mathbf F_{\boldsymbol\xi}]_{i,j}
=
2\operatorname{Re}
\left\{
\frac{\partial \boldsymbol\mu^H(\boldsymbol\xi)}
{\partial \xi_i}
\left(
\mathbf I_L\otimes
\mathbf R_{{\rm v},{\rm s}}^{-1}
\right)
\frac{\partial \boldsymbol\mu(\boldsymbol\xi)}
{\partial \xi_j}
\right\},
\label{eq:full_fim}
\end{equation}
where \(\mathbf I_L\otimes\mathbf R_{{\rm v},{\rm s}}\) denotes the
covariance matrix of the stacked effective interference-plus-noise vector
over \(L\) snapshots.

Since the reflection coefficients are nuisance parameters for the
considered angle-estimation task, we define
$\boldsymbol\eta
=
[
\operatorname{Re}\{\boldsymbol\beta\}^T,
\operatorname{Im}\{\boldsymbol\beta\}^T
]^T
\in\mathbb R^{2K_s\times 1}.$
Accordingly, the full FIM can be partitioned as
\begin{equation}
\mathbf F_{\boldsymbol\xi}
=
\begin{bmatrix}
\mathbf F_{\boldsymbol\theta\boldsymbol\theta}
&
\mathbf F_{\boldsymbol\theta\boldsymbol\eta}
\\
\mathbf F_{\boldsymbol\eta\boldsymbol\theta}
&
\mathbf F_{\boldsymbol\eta\boldsymbol\eta}
\end{bmatrix}.
\label{eq:fim_partition}
\end{equation}
After eliminating the nuisance reflection parameters via the Schur
complement, the equivalent FIM for the target angles is given by
$\mathbf J_{\boldsymbol\theta}
=
\mathbf F_{\boldsymbol\theta\boldsymbol\theta}
-
\mathbf F_{\boldsymbol\theta\boldsymbol\eta}
\mathbf F_{\boldsymbol\eta\boldsymbol\eta}^{-1}
\mathbf F_{\boldsymbol\eta\boldsymbol\theta}.$
Therefore, the angle-related CRB metric is defined as
\begin{equation}
{\rm CRB}_{\boldsymbol\theta}
=
\operatorname{tr}
\left(
\mathbf J_{\boldsymbol\theta}^{-1}
\right).
\label{eq:angle_crb_metric}
\end{equation}
}

\section{Sensing Performance Analysis and Problem Formulation}
\label{sec3}

In this section, we first provide an order-wise CRB analysis to quantify the sensing gain enabled by the co-prime array layout. Subsequently, we formulate a joint resource allocation problem for the considered TDD communication and FD sensing ISAC system, where the weighted downlink--uplink communication sum rate is maximized subject to the sensing accuracy, total BS transmit-power, downlink/uplink QoS, and residual SI constraints.

\subsection{Order-Wise CRB Analysis}

To obtain analytical insight into the sensing advantage of the proposed CPA-based architecture, we first characterize the  normalized correlation-response of the CPA-induced virtual array manifold. We then investigate the asymptotic  behavior (i.e., scaling order)  of the angle-related CRB with respect to the allocated antennas by comparing it with that of a conventional ULA benchmark. Similar normalized steering-vector correlation metrics have been widely adopted as ambiguity functions or virtual-array beampatterns in MIMO radar array characterization, where they quantify whether two angular directions can be distinguished from the array responses~\cite{7928056,10538290}.

For a reference direction $\theta_0$ and a scanning direction $\theta$, the normalized correlation response of the virtual array manifold is defined as
\begin{equation}
\chi(\theta_0,\theta)
=
\frac{
\left|
\boldsymbol{\phi}^{H}(\theta_0)
\boldsymbol{\phi}(\theta)
\right|
}{
\left\|
\boldsymbol{\phi}(\theta_0)
\right\|_2
\left\|
\boldsymbol{\phi}(\theta)
\right\|_2
}.
\label{eq:norm_corr}
\end{equation}
 The corresponding response is given by
$
20\log_{10}
\left(
\chi(\theta_0,\theta)
\right).$

\begin{figure}[t]
    \centering
    \includegraphics[width=0.92\columnwidth]{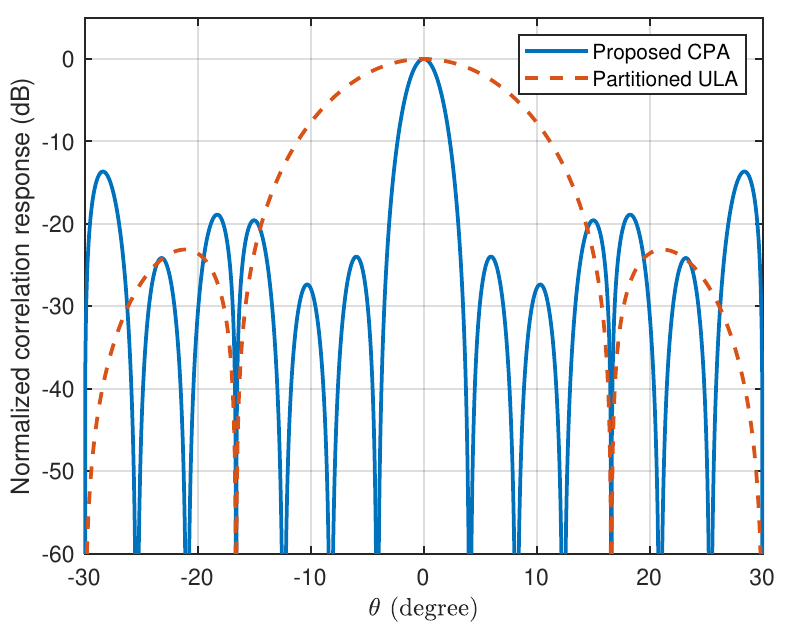}
    \caption{
Normalized correlation responses of the proposed CPA and the partitioned ULA benchmark with 4 sensing transmit antenna elements and 7 sensing receive antenna elements, i.e., $(M_1, M_2)=(4,7)$, for reference direction $\theta_0=0^\circ$.}
    \label{fig:norm_corr_response}
\end{figure}

As shown in Fig. 2, for the same numbers of sensing
transmit and receive antennas, the proposed CPA exhibits
a much narrower mainlobe than the partitioned ULA. This
narrow mainlobe from the CPA is due to the difference
between sensing and communications. This behavior is fundamentally due to the transmit-receive virtual aperture created by the CPA geometry.  In conventional communication-type beamforming ,  antenna array is mainly designed to improve the spatial degrees of
freedom for one-way signal transmission from the BS to users
or from users to the BS. Accordingly, half-wavelength-spaced ULAs
are commonly adopted to form directing  beams while avoiding
spatial aliasing. In contrast, sensing depends on the phase variations accumulated along both the probing path from the transmitter to the target and the echo path from the target to the receiver. Hence, the relevant array response is not determined by the transmit beam directing.
Therefore, the virtual array responses
decorrelate more rapidly as the scanning direction deviates
from the reference direction, which indicates a higher angular
resolution capability.
% As shown in Fig.~\ref{fig:norm_corr_response}, for the same numbers of sensing
% transmit and receive antennas, the proposed CPA exhibits
% a much narrower mainlobe than the partitioned ULA. This
% narrow mainlobe is due to the fundamental difference between
% communication and sensing array responses. For communication,
% the antenna array is mainly designed to improve the spatial
% degrees of freedom for one-way signal transmission from the
% BS to the UEs or from the UEs to the BS, and a
% half-wavelength-spaced ULA is commonly adopted to form
% narrow directional beams while avoiding spatial aliasing.
% However, sensing includes both the phase accumulation of
% the transmitted probing signal toward the target and that of
% the reflected echo received from the target. Hence, the sensing
% angular response is governed by the transmit--receive virtual
% array manifold. The co-prime transmit and receive
% spacings enlarge the effective virtual aperture and make the
% virtual array response decorrelate more rapidly with angular
% deviation, thereby producing the narrower mainlobe. 

Motivated by the narrower correlation mainlobe observed above, we next investigate the order-wise behavior of the angle-related CRB and compare it with that of a ULA benchmark. Specifically, we begin with the single-target case, where the asymptotic scaling can be directly characterized from the derivative structure of the CPA virtual manifold. We then extend the analysis to the multi-target case. This extension is nontrivial, since, unlike the single-target setting, the multi-target CRB depends on the effective angle Fisher information matrix after eliminating the nuisance parameters associated with the target reflection coefficients. Therefore, the multi-target analysis must retain the complete FIM block structure and its Schur-complement correction, under which a nondegenerate target configuration enables a meaningful order-wise comparison between the CPA and the ULA.

To facilitate the following asymptotic analysis, we introduce the spatial frequency $\omega(\theta)=\pi\sin\theta$ and reformulate the angle-related CRB with respect to $\omega$, since $\omega$ leads to a more compact characterization of the steering vectors and their derivatives. Further, we eliminate the nuisance parameters from the full mean-based FIM. For the single-target case, define the parameter vector as
\begin{equation}
\label{eq:vector}
\boldsymbol{\xi}
=
\begin{bmatrix}
\omega \\
\boldsymbol{\eta}
\end{bmatrix},
\qquad
\boldsymbol{\eta}
=
\begin{bmatrix}
\Re\{\beta\} \\
\Im\{\beta\}
\end{bmatrix},
\end{equation}
where $\omega$ is the spatial frequency of interest and $\boldsymbol{\eta}$ collects the nuisance parameters associated with the complex reflection coefficient. Then, the corresponding FIM can be partitioned as
\begin{equation}
\mathbf{F}
=
\begin{bmatrix}
F_{\omega\omega} & \mathbf{F}_{\omega\eta} \\
\mathbf{F}_{\eta\omega} & \mathbf{F}_{\eta\eta}
\end{bmatrix}.
\end{equation}
Since the CRB of the spectral frequency $\omega$ is given by the $(1,1)$-th block of $\mathbf{F}^{-1}$, applying the block matrix inversion formula yields
\begin{equation}
\mathrm{CRB}_{\omega}
=
\left(
F_{\omega\omega}
-
\mathbf{F}_{\omega\eta}
\mathbf{F}_{\eta\eta}^{-1}
\mathbf{F}_{\eta\omega}
\right)^{-1}.
\label{eq:effective_CRB_single}
\end{equation}
The term
$F_{\omega\omega}
-
\mathbf{F}_{\omega\eta}
\mathbf{F}_{\eta\eta}^{-1}
\mathbf{F}_{\eta\omega}$
represents the effective Fisher information associated with the angle parameter after eliminating the nuisance parameters via the Schur complement. Therefore, the order-wise CRB analysis can be equivalently carried out by studying the asymptotic scaling of $F_{\omega\omega}
-
\mathbf{F}_{\omega\eta}
\mathbf{F}_{\eta\eta}^{-1}
\mathbf{F}_{\eta\omega}$. As for the multi-target case, the same principle applies, except that the scalar quantity in \eqref{eq:effective_CRB_single} is replaced by the effective angle FIM matrix obtained from the corresponding block Schur complement.

\begin{theorem}[Single-target order-wise CRB comparison]
Consider the single-target sensing case. Let the CPA sensing layout
be generated by the co-prime pair \((M_1,M_2)\) satisfying
\(M_s=M_1+M_2-1\), \(M_1=\rho(M_s+1)\), and
\(M_2=(1-\rho)(M_s+1)\), where \(\rho\in(0,1)\) is chosen such that
\(M_1\) and \(M_2\) are co-prime positive integers. The transmit power of
each sensing transmit antenna is kept fixed as \(M_s\) grows. Then,
as \(M_s\to\infty\), the CRB of the spatial frequency \(\omega\)
for the proposed CPA-based sensing architecture satisfies
\begin{equation}
\mathrm{CRB}_{\omega,\mathrm{CPA}}=O(M_s^{-6}).
\end{equation}
For the partitioned ULA sensing benchmark with the same number of
sensing elements, the corresponding CRB satisfies
\begin{equation}
    \mathrm{CRB}_{\omega,\mathrm{ULA}}=O(M_s^{-4}).
\end{equation}
\end{theorem}

\begin{proof}
Please refer to Appendix A.
\end{proof}

% \begin{theorem}
% \label{thm:CPA_dominant_term}

% As $M_s$ increases for a fixed sensing power budget, the angle estimation exhibits an order-wise decay of $O(M_s^{-6})$ with the proposed CPA-based sensing architecture, whereas it decays as $O(M_s^{-4})$ for the partitioned ULA sensing baseline. For CPA sensing, we have
% \begin{equation} 
% \begin{aligned}
% \dot{\boldsymbol{\phi}}^{H}(\omega)\dot{\boldsymbol{\phi}}(\omega)
% &\approx
% \frac{7}{6}\rho^{3}(1-\rho)^{3}(M_s+1)^{6}.
% \end{aligned}
% \end{equation}
% \end{theorem}

% \begin{proof}
% See Appendix~\ref{app:proof_thm1}.
% \end{proof}

\begin{figure}[t]
    \centering
    \includegraphics[width=0.92\columnwidth]{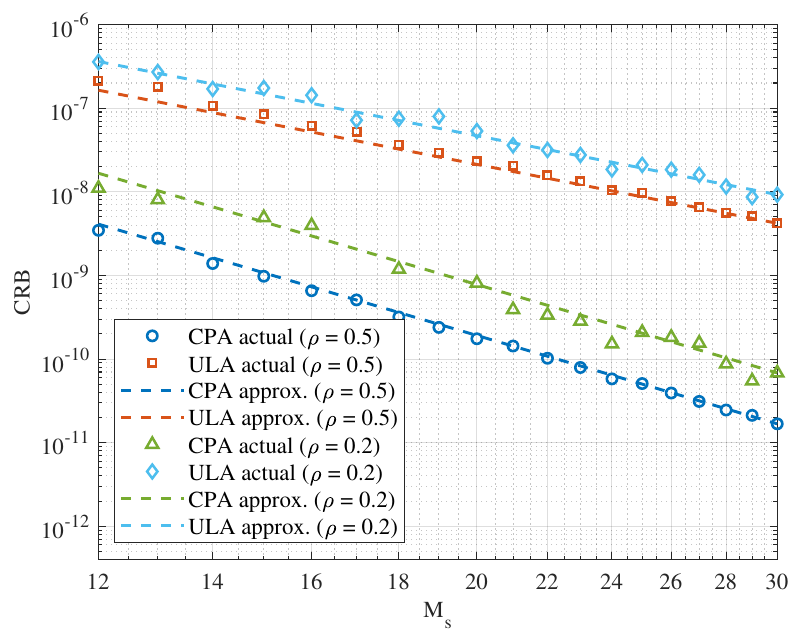}
    \caption{CRB versus the number of sensing elements $M_s$. The dashed lines indicate the asymptotic reference slopes predicted by Theorem~1, namely, CPA $\propto M_s^{-6}$ and ULA $\propto M_s^{-4}$.}
    \label{fig:crb_vs_ms_single_target}
\end{figure}

% Theorem 1 illustrates that under the same number of physical sensing elements and a fixed sensing power budget, the proposed CPA-based sensing architecture enjoys a faster asymptotic CRB decay than the partitioned ULA benchmark.

Theorem 1 illustrates that, under the same number of physical
sensing elements and with the transmit power of each sensing
transmit antenna kept fixed, the proposed CPA-based sensing
architecture enjoys a faster asymptotic CRB decay than the
partitioned ULA benchmark. This result reflects the virtual-aperture gain of the
CPA geometry in the order-wise array-geometry comparison.

% \begin{remark}
% % Theorem~\ref{thm:CPA_dominant_term} quantifies the leading-order scaling
% % $\dot{\boldsymbol{\phi}}^{H}(\theta)\dot{\boldsymbol{\phi}}(\theta)$, and hence that of $F_{\theta\theta}(m,m)$ under \eqref{eq:jiabi}.
% % For the same number of sensing antenna elements, the $O(M_s^{2})$ order gain explains the CRB advantage of the CPA-induced over conventional ULA as the increase of $M_s$.
% \end{remark}

On the other hand, to characterize the asymptotic CRB behavior in the multi-target case, we first introduce the following nondegenerate target placement.

\begin{definition}
Consider the multi-target sensing case with target spatial-frequency vector
$\boldsymbol{\omega} = [\omega_1,\ldots,\omega_{K_{\mathrm{s}}}]^T$, where $\omega_m = \pi \sin \theta_m$.
The target placement is said to be nondegenerate if there exists a constant
$\delta \in (0,1)$, independent of $M_{\mathrm{s}}$, such that for any $m \neq n$,
$\left|\sin\!\left(\frac{M_1(\omega_m-\omega_n)}{2}\right)\right| \ge \delta$
and
$\left|\sin\!\left(\frac{M_2(\omega_m-\omega_n)}{2}\right)\right| \ge \delta$.
\end{definition}

This condition excludes degenerate target configurations where
the pairwise spatial-frequency differences fall close to the
aliasing singularities induced by the co-prime spacings. Under Definition 1, the inter-target cross terms are asymptotically weaker than the corresponding diagonal terms after the Schur-complement reduction.

\begin{theorem}
% Under the same CPA sensing configuration as in Theorem 1, consider the multi-target sensing case with $K_{\mathrm{s}} \ge 2$. If the target placement satisfies Definition 1, then, under a fixed sensing power budget, the trace of the multi-target spatial-frequency CRB obtained from the full Fisher information matrix after eliminating the nuisance reflection parameters via the block Schur complement satisfies
% \begin{equation}
% \mathrm{tr}\!\left(\mathbf{CRB}_{\boldsymbol{\omega},\mathrm{CPA}}\right)=O(M_{\mathrm{s}}^{-6})
% \end{equation}
% for sufficiently large $M_{\mathrm{s}}$.

% For the partitioned ULA sensing benchmark with the same number of sensing elements, under the corresponding nondegenerate target placement, the trace of the multi-target spatial-frequency CRB satisfies
% \begin{equation}
% \mathrm{tr}\!\left(\mathbf{CRB}_{\boldsymbol{\omega},\mathrm{ULA}}\right)=O(M_{\mathrm{s}}^{-4})
% \end{equation}
% for sufficiently large $M_{\mathrm{s}}$.

Under the same CPA sensing configuration as in Theorem 1,
consider the multi-target sensing case with fixed \(K_s\ge 2\).
If the target placement satisfies Definition 1, then, with the transmit power of each sensing transmit antenna kept fixed
as in Theorem 1, the trace of the multi-target spatial-frequency CRB
obtained from the full Fisher information matrix after eliminating
the nuisance reflection parameters via the block Schur complement
satisfies
\begin{equation}
    \operatorname{tr}(\mathrm{CRB}_{\omega,\mathrm{CPA}})
=
O(M_s^{-6})
\end{equation}
for sufficiently large \(M_s\).

For the partitioned ULA sensing benchmark with the same number
of sensing elements, under the corresponding nondegenerate target
placement and the same antenna-power setting, the trace of the
multi-target spatial-frequency CRB satisfies
\begin{equation}
\operatorname{tr}(\mathrm{CRB}_{\omega,\mathrm{ULA}})
=
O(M_s^{-4})
\end{equation}
for sufficiently large \(M_s\).
\end{theorem}

\begin{proof}
Please refer to Appendix B.
\end{proof}

Building on Theorems 1 and 2, we further characterize the snapshot requirement under a fixed physical aperture. Specifically, although the proposed CPA can realize a much larger effective aperture with fewer physical sensing elements, this sparse spatial sampling also leads to an estimation-accuracy loss relative to a same-aperture ULA, which must be compensated by using more temporal snapshots.

\begin{theorem}[Snapshot tradeoff under the same physical sensing aperture]
% Under the same physical aperture, let $L_{\mathrm{CPA}}$ and $L_{\mathrm{ULA}}$ denote the numbers of snapshots required by the proposed CPA sensing architecture and the same-aperture ULA benchmark, respectively. Then, under a fixed sensing power budget, attaining the same spatial-frequency CRB generally requires the proposed CPA to exploit more snapshots than the ULA benchmark.

Under the same physical sensing aperture, let \(L_{\mathrm{CPA}}\)
and \(L_{\mathrm{ULA}}\) denote the numbers of snapshots required by
the proposed CPA sensing architecture and the same-aperture ULA
benchmark, respectively. With the transmit power of each sensing transmit antenna kept fixed
as in the above order-wise analysis, attaining the
same spatial-frequency CRB generally requires the proposed CPA to
exploit more snapshots than the ULA benchmark. 
\end{theorem}

\begin{proof}
From the proofs of Theorems 1 and 2 under the same fixed
antenna-power setting, the single-target
CRB and the multi-target trace-CRB both scale inversely with the
number of snapshots. Hence, for fixed \(K_s\), the proposed CPA
and the same-aperture ULA benchmark satisfy
$\mathrm{CRB}_{\omega,\mathrm{CPA}}
=
\Theta\!\left(\frac{1}{L_{\mathrm{CPA}} M_{\mathrm{s}}^{6}}\right),
\mathrm{CRB}_{\omega,\mathrm{ULA}}
=
\Theta\!\left(\frac{1}{L_{\mathrm{ULA}} M_{\mathrm{U}}^{4}}\right),$
where $M_{\mathrm{U}}$ denotes the number of sensing elements of the ULA benchmark.

For the proposed CPA, the physical sensing aperture is determined by the outermost deployed sensing element  and thus satisfies
\begin{equation}
A_{\mathrm{CPA}}=\Theta(M_1 M_2 d)=\Theta(M_{\mathrm{s}}^{2} d),
\end{equation}
since $M_1=\rho(M_{\mathrm{s}}+1)$ and $M_2=(1-\rho)(M_{\mathrm{s}}+1)$ imply $M_1M_2=\Theta(M_{\mathrm{s}}^{2})$. For a ULA with $M_{\mathrm{U}}$ sensing elements, its physical aperture satisfies
\begin{equation}
A_{\mathrm{ULA}}=(M_{\mathrm{U}}-1)d=\Theta(M_{\mathrm{U}} d).
\end{equation}
Therefore, under the same physical aperture condition $A_{\mathrm{CPA}}\asymp A_{\mathrm{ULA}}$, we have
$M_{\mathrm{U}}=\Theta(M_{\mathrm{s}}^{2}).$

Substituting this relation into the ULA CRB scaling gives
\begin{equation}
\mathrm{CRB}_{\omega,\mathrm{ULA}}
=
\Theta\!\left(\frac{1}{L_{\mathrm{ULA}} M_{\mathrm{s}}^{8}}\right).
\end{equation}
To attain the same spatial-frequency CRB, it is therefore necessary that
$\Theta\!\left(\frac{1}{L_{\mathrm{CPA}} M_{\mathrm{s}}^{6}}\right)
=
\Theta\!\left(\frac{1}{L_{\mathrm{ULA}} M_{\mathrm{s}}^{8}}\right),$
which yields
$L_{\mathrm{CPA}}=\Theta(M_{\mathrm{s}}^{2})L_{\mathrm{ULA}}.$
Hence, under the same physical aperture, the proposed CPA generally requires more snapshots than the ULA benchmark to attain the same spatial-frequency CRB. This proves Theorem~3.
\end{proof}

% \begin{proof}

% From Theorems 1 and 2, and by retaining the snapshot factor explicitly in
% the effective Fisher information, the CPA and ULA CRBs scale as
% \[
% \mathrm{CRB}_{\omega,\mathrm{CPA}}
% =
% \Theta\!\left(\frac{1}{L_{\mathrm{CPA}} M_s^6}\right),
% \qquad
% \mathrm{CRB}_{\omega,\mathrm{ULA}}
% =
% \Theta\!\left(\frac{1}{L_{\mathrm{ULA}} M_{\mathrm{U}}^4}\right),
% \]
% where $M_{\mathrm{U}}$ denotes the number of sensing elements of the ULA.
% For the proposed CPA, the sensing aperture is determined by the largest
% sensor location and hence satisfies
% \[
% A_{\mathrm{CPA}}=\Theta(M_1M_2 d)=\Theta(M_s^2 d),
% \]
% since $M_1=\rho(M_s+1)$ and $M_2=(1-\rho)(M_s+1)$ imply
% $M_1M_2=\Theta(M_s^2)$. For a ULA with $M_{\mathrm{U}}$ sensing elements,
% its physical aperture satisfies
% \[
% A_{\mathrm{ULA}}=(M_{\mathrm{U}}-1)d=\Theta(M_{\mathrm{U}}d).
% \]
% Therefore, under the same physical aperture condition
% $A_{\mathrm{ULA}}\asymp A_{\mathrm{CPA}}$, we have
% \[
% M_{\mathrm{U}}=\Theta(M_s^2).
% \]
% Substituting this relation into the ULA CRB expression yields
% \[
% \mathrm{CRB}_{\omega,\mathrm{ULA}}
% =
% \Theta\!\left(\frac{1}{L_{\mathrm{ULA}} M_s^8}\right).
% \]
% To attain the same order-wise spatial-frequency CRB, we require
% \[
% \Theta\!\left(\frac{1}{L_{\mathrm{CPA}} M_s^6}\right)
% =
% \Theta\!\left(\frac{1}{L_{\mathrm{ULA}} M_s^8}\right),
% \]
% which immediately gives
% \[
% L_{\mathrm{CPA}}=\Theta(M_s^2)L_{\mathrm{ULA}}.
% \]
% \end{proof}

To validate the order-wise CRB scaling in Theorem~1,
Fig.~\ref{fig:crb_vs_ms_single_target} illustrates the CRB changes with the number of sensing antenna elements $M_s$ for different $\rho$. The results demonstrate  that, in the large-$M_s$ regime, the CPA and ULA curves approach the theoretical orders $M_s^{-6}$ and $M_s^{-4}$, respectively. Meanwhile, when $\rho=0.5$, where the number of transmit antennas is equal to the number of receive antennas, the CRB is less than the same number of sensing antenna elements with  $\rho=0.2$. This is because, for a fixed $M_s$, $\rho=0.5$ yields $M_1\approx M_2$, which maximizes the product $M_1M_2$ and thereby
generates a richer CPA-induced virtual array with a larger effective
aperture. Consequently, the balanced Tx/Rx configuration offers improved
angular resolution and lower CRB.

{\color{black}

\subsection{Problem Formulation}

In this subsection, we formulate a joint resource allocation problem for the considered TDD communication and FD sensing ISAC system. The objective is to maximize the weighted 
communication sum rate by jointly designing the sensing
precoder, the downlink communication precoder, and the uplink
receive beamformers, subject to the sensing accuracy
requirement, BS transmit-power budget, residual SI constraints,
and communication QoS constraints. The optimization problem is formulated as
\begin{equation}
\small
\begin{aligned}
\mathbf P1:\quad
&\max_{\mathbf W_s,\mathbf W_c^{\rm DL},\mathbf W_c^{\rm UL}}
\quad
\alpha_{\rm DL}R_{\rm DL}
+
\alpha_{\rm UL}R_{\rm UL}
\\
\mathrm{s.t.}\quad
\mathrm{C1}:&
\quad
\operatorname{CRB}_{\theta}(\mathbf W_s)
\le
\Gamma_{\rm CRB},
\\
\mathrm{C2}:&
\quad
\|\mathbf W_s\|_F^2
+
\|\mathbf W_c^{\rm DL}\|_F^2
\le
P_{\rm BS}^{\max},
\\
\mathrm{C3}:&
\quad
\left[
\mathbf H_{{\rm SI},s}
\mathbf R_s(\mathbf W_s)
\mathbf H_{{\rm SI},s}^H
\right]_{m,m}
\le
\Gamma_{{\rm SI},s},
\\
&\quad
m=1,\ldots,M_2,
\\
\mathrm{C4}:&
\quad
\left[
\mathbf H_{{\rm SI},c}
\mathbf R_s(\mathbf W_s)
\mathbf H_{{\rm SI},c}^H
\right]_{q,q}
\le
\Gamma_{{\rm SI},c},
\\
&\quad
q=1,\ldots,M_c,
\\
\mathrm{C5}:&
\quad
\gamma_k^{\rm DL}
\ge
\gamma_{\rm DL}^{\rm th},
\quad
k=1,\ldots,K_c,
\\
\mathrm{C6}:&
\quad
\gamma_k^{\rm UL}
\ge
\gamma_{\rm UL}^{\rm th},
\quad
k=1,\ldots,K_c,
\\
\mathrm{C7}:&
\quad
\|\mathbf w_{c,k}^{\rm UL}\|^2=1,
\quad
k=1,\ldots,K_c.
\end{aligned}
\end{equation}
Constraint C1 guarantees the required sensing accuracy by
limiting the angle-related CRB. Constraint C2 imposes the total
BS transmit-power budget for the dedicated sensing waveform
and the downlink communication signals. Constraints C3 and C4
restrict the residual SI power at the sensing receive subarray
and the communication transceiver antennas, respectively.
Constraints C5 and C6 ensure the downlink and uplink
communication QoS requirements, respectively. Constraint C7
normalizes the uplink receive beamformers.

Problem $\mathbf P1$ is non-convex due to the coupled
communication sum-rate objective, the fractional
SINR expressions in both downlink and uplink links, and the
CRB constraint. In particular, the sensing precoder
$\mathbf W_s$ affects the downlink rate through the sensing
interference term, the uplink rate through the residual SI term,
and the sensing accuracy through the induced covariance
$\mathbf R_s(\mathbf W_s)$. Moreover, the downlink precoder
$\mathbf W_c^{\rm DL}$ is coupled with the downlink SINR
constraints and the total BS transmit-power budget. Therefore,
a tractable iterative algorithm is required, which will be
developed in the following section.
}

{\color{black}
\section{Problem Solution}\label{sec4}
To tackle the non-convex problem $\mathbf P1$, we first
reformulate the downlink and uplink rate terms using
fractional programming (FP). Based on the resulting equivalent
formulation, an alternating optimization (AO) algorithm is
developed to successively update the FP auxiliary variables,
the communication beamformers, and the sensing transmit
covariance until convergence.

\subsection{FP-Based Rate Reformulation}

Let $q\in\{\mathrm{DL},\mathrm{UL}\}$ denote the communication
direction, where $\alpha_q$ is the corresponding downlink or
uplink rate weight. The desired signal amplitudes are defined as
$a_k^{\rm DL}
=
(\mathbf h_{c,k}^{\rm DL})^H
\mathbf w_{c,k}^{\rm DL},
\qquad
a_k^{\rm UL}
=
\sqrt{p_k^{\rm UL}}
(\mathbf w_{c,k}^{\rm UL})^H
\mathbf h_{c,k}^{\rm UL}.$
The corresponding total received powers are
\begin{align}
\small
\mathcal I_k^{\rm DL}
={}&
\sum_{j=1}^{K_c}
\left|
(\mathbf h_{c,k}^{\rm DL})^H
\mathbf w_{c,j}^{\rm DL}
\right|^2
+
(\mathbf h_{s,k}^{\rm DL})^H
\mathbf R_s
\mathbf h_{s,k}^{\rm DL}
+
\sigma_{\rm DL}^2,
\\
\mathcal I_k^{\rm UL}
={}&
\sum_{j=1}^{K_c}
p_j^{\rm UL}
\left|
(\mathbf w_{c,k}^{\rm UL})^H
\mathbf h_{c,j}^{\rm UL}
\right|^2
\nonumber\\
&+
(\mathbf w_{c,k}^{\rm UL})^H
\mathbf H_{{\rm SI},c}
\mathbf R_s
\mathbf H_{{\rm SI},c}^H
\mathbf w_{c,k}^{\rm UL}
+
\sigma_{\rm UL}^2
\|\mathbf w_{c,k}^{\rm UL}\|^2.
\end{align}
Accordingly, the
SINR can be uniformly expressed as
\begin{equation}
    \gamma_k^q
=
\frac{|a_k^q|^2}
{\mathcal I_k^q-|a_k^q|^2},
q\in\{\mathrm{DL},\mathrm{UL}\}.
\end{equation}

Following the FP-based rate transformation in~\cite{11168854},
the weighted communication sum rate can be equivalently expressed as
\begin{equation}
R_{\rm sum}
=
\max_{\{\nu_k^q\geq0,\,y_k^q\in\mathbb C\}}
\widetilde R_{\rm sum},
\end{equation}
where
\begin{align}
\widetilde R_{\rm sum}
&=
\frac{B}{\ln 2}
\sum_{q\in\{\mathrm{DL},\mathrm{UL}\}}
\alpha_q
\sum_{k=1}^{K_c}
\Bigg[
\ln(1+\nu_k^q)-\nu_k^q
\nonumber\\
&+
2\operatorname{Re}
\left\{
(y_k^q)^*
\sqrt{1+\nu_k^q}\,
a_k^q
\right\}
-
|y_k^q|^2\mathcal I_k^q
\Bigg].
\end{align}
For given physical design variables, the optimal auxiliary
variables are updated in closed form as
\begin{equation}
(\nu_k^q)^\star
=
\gamma_k^q,
(y_k^q)^\star
=
\frac{
\sqrt{1+(\nu_k^q)^\star}\,a_k^q
}{
\mathcal I_k^q
},
q\in\{\mathrm{DL},\mathrm{UL}\}.
\label{eq:fp_auxiliary_update}
\end{equation}
The uplink receive beamformers admit closed-form updates,
while the transformed objective is used to construct tractable
subproblems for the downlink precoder and the sensing transmit
covariance.

\subsection{Communication Beamforming Optimization with Fixed Sensing Transmit Design}
For fixed $\mathbf W_s$, the communication beamformers are
updated by first determining the uplink receive beamformers in
closed form and then optimizing the downlink precoder.

For the uplink link, define the interference-plus-noise covariance
matrix of the $k$-th communication UAV as
\begin{equation}
\mathbf C_k^{\rm UL}
=
\sum_{j\ne k}^{K_c}
p_j^{\rm UL}
\mathbf h_{c,j}^{\rm UL}
(\mathbf h_{c,j}^{\rm UL})^H
+
\mathbf H_{{\rm SI},c}
\mathbf R_s(\mathbf W_s)
\mathbf H_{{\rm SI},c}^H
+
\sigma_{\rm UL}^2\mathbf I_{M_c}.
\end{equation}
The SINR-maximizing receive beamformer is given by
\begin{equation}
(\mathbf w_{c,k}^{\rm UL})^\star
=
\frac{
(\mathbf C_k^{\rm UL})^{-1}
\mathbf h_{c,k}^{\rm UL}
}{
\left\|
(\mathbf C_k^{\rm UL})^{-1}
\mathbf h_{c,k}^{\rm UL}
\right\|
},
\quad
k=1,\ldots,K_c.
\label{eq:ul_beamformer_update}
\end{equation}
The above normalization satisfies $\mathrm{C7}$, while the
SINR-maximizing update preserves $\mathrm{C6}$ for a feasible
current iterate.

For the downlink link, define the fixed
sensing-interference-plus-noise power as
\begin{equation}
\zeta_k^{\rm DL}
=
(\mathbf h_{s,k}^{\rm DL})^H
\mathbf R_s(\mathbf W_s)
\mathbf h_{s,k}^{\rm DL}
+
\sigma_{\rm DL}^2.
\end{equation}

Without loss of optimality, the phase of
$\mathbf w_{c,k}^{\rm DL}$ can be chosen such that
$(\mathbf h_{c,k}^{\rm DL})^H\mathbf w_{c,k}^{\rm DL}$
is real and nonnegative~\cite{10158711}. Accordingly, constraint
$\mathrm{C5}$ is equivalently recast as the SOC constraint
\begin{equation}
\small
\begin{aligned}
\widetilde{\mathrm{C5}}:\quad
&\left\|
\left[
\left\{
(\mathbf h_{c,k}^{\rm DL})^H
\mathbf w_{c,j}^{\rm DL}
\right\}_{j\ne k},
\sqrt{\zeta_k^{\rm DL}}
\right]^T
\right\|_2
\le
\frac{
\operatorname{Re}
\left\{
(\mathbf h_{c,k}^{\rm DL})^H
\mathbf w_{c,k}^{\rm DL}
\right\}
}{
\sqrt{\gamma_{\rm DL}^{\rm th}}
},
\\
&k=1,\ldots,K_c.
\end{aligned}
\end{equation}

For fixed FP auxiliary variables, the downlink precoder is updated
by solving
\begin{equation}
\small
\begin{aligned}
\mathbf P2:\quad
\max_{\mathbf W_c^{\rm DL}}\quad
&
\frac{\alpha_{\rm DL}B}{\ln 2}
\sum_{k=1}^{K_c}
\Bigg[
2\operatorname{Re}
\left\{
(y_k^{\rm DL})^*
\sqrt{1+\nu_k^{\rm DL}}\,
(\mathbf h_{c,k}^{\rm DL})^H
\mathbf w_{c,k}^{\rm DL}
\right\}
\\
&\hspace{21mm}
-
|y_k^{\rm DL}|^2
\mathcal I_k^{\rm DL}
\Bigg]
\\
\mathrm{s.t.}\quad
&
\widetilde{\mathrm{C2}}:\ 
\left\|
\mathbf W_c^{\rm DL}
\right\|_F^2
\le
P_{\rm BS}^{\max}
-
\left\|
\mathbf W_s
\right\|_F^2,
\\
&
\widetilde{\mathrm{C5}}.
\end{aligned}
\label{prob:downlink_subproblem}
\end{equation}
For fixed FP auxiliary variables, the objective of $\mathbf P2$
is concave in $\mathbf W_c^{\rm DL}$, while
$\widetilde{\mathrm{C2}}$ and $\widetilde{\mathrm{C5}}$ define a
convex feasible set. Hence, $\mathbf P2$ is a convex
quadratically constrained problem and can be efficiently solved
using standard convex optimization tools.

\subsection{Sensing Transmit Design Optimization with Fixed
Communication Beamformers}

For fixed communication beamformers
$\{\mathbf W_c^{\rm DL},\mathbf W_c^{\rm UL}\}$, we introduce the
lifted sensing covariance
$\mathbf R_s=\mathbf W_s\mathbf W_s^H\succeq\mathbf 0$.
Define
$\mathbf G_{s,k}^{\rm DL}
=
\mathbf h_{s,k}^{\rm DL}
(\mathbf h_{s,k}^{\rm DL})^H,
\mathbf G_{{\rm SI},k}^{\rm UL}
=
\mathbf H_{{\rm SI},c}^H
\mathbf w_{c,k}^{\rm UL}
(\mathbf w_{c,k}^{\rm UL})^H
\mathbf H_{{\rm SI},c}.$

For fixed FP auxiliary variables, define
$\lambda_k^{\rm DL}
=
\alpha_{\rm DL}|y_k^{\rm DL}|^2$
and
$\lambda_k^{\rm UL}
=
\alpha_{\rm UL}|y_k^{\rm UL}|^2$.
After omitting the common positive factor $B/\ln 2$ and the
terms independent of $\mathbf R_s$, maximizing the transformed
weighted sum-rate is equivalent to minimizing the weighted
sensing-interference and residual-SI terms.

With fixed communication beamformers, the downlink and uplink
QoS constraints can be equivalently rewritten as
\begin{align}
\small
\overline{\mathrm{C5}}:\quad
\mathcal I_k^{\rm DL}
&\le
\left(
1+\frac{1}{\gamma_{\rm DL}^{\rm th}}
\right)
|a_k^{\rm DL}|^2,
\quad
k=1,\ldots,K_c,
\\
\overline{\mathrm{C6}}:\quad
\mathcal I_k^{\rm UL}
&\le
\left(
1+\frac{1}{\gamma_{\rm UL}^{\rm th}}
\right)
|a_k^{\rm UL}|^2,
\quad
k=1,\ldots,K_c.
\end{align}

The total transmit-power constraint becomes
\begin{equation}
\overline{\mathrm{C2}}:\quad
\operatorname{tr}(\mathbf R_s)
\le
P_{\rm BS}^{\max}
-
\left\|
\mathbf W_c^{\rm DL}
\right\|_F^2.
\end{equation}

To obtain a tractable CRB constraint, we adopt a fixed
worst-case upper bound
$\overline{\mathbf R}_{{\rm SI},s}$ on the residual-SI
covariance and define
\begin{equation}
\overline{\mathbf R}_{v,s}
=
\mathbf R_{\ell,s}
+
\overline{\mathbf R}_{{\rm SI},s}
+
\sigma_s^2\mathbf I_{M_2}.
\end{equation}
With $\overline{\mathbf R}_{v,s}$ fixed, the resulting FIM
$\overline{\mathbf F}_{\boldsymbol\xi}(\mathbf R_s)$ is affine
in $\mathbf R_s$.

By introducing an auxiliary matrix $\mathbf T_\theta$, the CRB
constraint is conservatively reformulated as
\begin{align}
\small
\overline{\mathrm{C1a}}:\quad
\begin{bmatrix}
\overline{\mathbf F}_{\boldsymbol\xi}(\mathbf R_s)
&
\mathbf E_\theta
\\
\mathbf E_\theta^T
&
\mathbf T_\theta
\end{bmatrix}
&\succeq\mathbf 0,
\\
\overline{\mathrm{C1b}}:\quad
\operatorname{tr}(\mathbf T_\theta)
&\le
\Gamma_{\rm CRB},
\end{align}
where $\mathbf E_\theta$ selects the angle-related parameters
from the complete parameter vector.

Since $\mathbf W_s\in\mathbb C^{M_1\times d_s}$, the lifted
covariance implicitly satisfies
$\operatorname{rank}(\mathbf R_s)\leq d_s$.
By dropping this non-convex rank constraint, the sensing
transmit design subproblem is relaxed as
\begin{equation}
\small
\begin{aligned}
\mathbf P3:\quad
\min_{\mathbf R_s,\mathbf T_\theta}\quad
&
\sum_{k=1}^{K_c}
\lambda_k^{\rm DL}
\operatorname{tr}
\left(
\mathbf G_{s,k}^{\rm DL}\mathbf R_s
\right)
\\
&+
\sum_{k=1}^{K_c}
\lambda_k^{\rm UL}
\operatorname{tr}
\left(
\mathbf G_{{\rm SI},k}^{\rm UL}\mathbf R_s
\right)
\\
\mathrm{s.t.}\quad
&
\overline{\mathrm{C1a}},\
\overline{\mathrm{C1b}},\
\overline{\mathrm{C2}},\
\mathrm{C3},\
\mathrm{C4},
\\
&
\overline{\mathrm{C5}},\
\overline{\mathrm{C6}},\
\mathbf R_s\succeq\mathbf 0.
\end{aligned}
\label{prob:sensing_subproblem}
\end{equation}
Problem $\mathbf P3$ is an SDP obtained through semidefinite
relaxation and can be efficiently solved using standard convex
optimization tools.
Let $\mathbf R_s^\star$ denote the solution of $\mathbf P3$.
If $\operatorname{rank}(\mathbf R_s^\star)\leq d_s$, an exact
sensing precoder can be recovered from its compact eigenvalue
decomposition; otherwise, Gaussian randomization is employed
to obtain a feasible rank-$d_s$ sensing precoder.
The overall AO procedure for solving $\mathbf P1$ is summarized
in Algorithm~\ref{alg:proposed_AO}.

\begin{algorithm}[t]
\small
\caption{AO-Based Joint Communication and Sensing Design}
\label{alg:proposed_AO}
\begin{algorithmic}[1]
\REQUIRE System parameters, feasible initial points
$\mathbf R_s^{(0)}$, $\mathbf W_c^{{\rm DL},(0)}$, and
$\mathbf W_c^{{\rm UL},(0)}$, convergence tolerance
$\epsilon$.
\ENSURE Optimized
$\mathbf R_s^\star$, $\mathbf W_c^{{\rm DL},\star}$, and
$\mathbf W_c^{{\rm UL},\star}$.
\STATE Set $t=0$ and calculate $R_{\rm sum}^{(0)}$.
\REPEAT
\STATE Update
$\{\nu_k^{q,(t+1)},y_k^{q,(t+1)}\}$,
$\forall k$ and $q\in\{{\rm DL},{\rm UL}\}$,
using~\eqref{eq:fp_auxiliary_update}.
\STATE Update $\mathbf W_c^{{\rm UL},(t+1)}$
using~\eqref{eq:ul_beamformer_update} with
$\mathbf R_s^{(t)}$.
\STATE Solve~\eqref{prob:downlink_subproblem} to obtain
$\mathbf W_c^{{\rm DL},(t+1)}$.
\STATE Solve~\eqref{prob:sensing_subproblem} to obtain
$\mathbf R_s^{(t+1)}$.
\STATE Calculate $R_{\rm sum}^{(t+1)}$ and set
$t\leftarrow t+1$.
\UNTIL{
$\displaystyle
\frac{
\left|R_{\rm sum}^{(t)}-R_{\rm sum}^{(t-1)}\right|
}{
\left|R_{\rm sum}^{(t-1)}\right|
}
\leq\epsilon$.
}
\STATE Set
$\mathbf R_s^\star=\mathbf R_s^{(t)}$,
$\mathbf W_c^{{\rm DL},\star}
=\mathbf W_c^{{\rm DL},(t)}$, and
$\mathbf W_c^{{\rm UL},\star}
=\mathbf W_c^{{\rm UL},(t)}$.
\end{algorithmic}
\end{algorithm}
}

The per-iteration complexity is dominated by the updates of the
communication beamformers and the sensing covariance. The uplink beamforming update requires
$\mathcal O(K_cM_c^3)$ operations, while $\mathbf P2$ and
$\mathbf P3$ are solved as an SOCP and an SDP, respectively.
Hence, the overall complexity scales as
$\mathcal O\!\left(
I_{\rm AO}
[K_cM_c^3+\mathcal C_{\rm SOCP}+\mathcal C_{\rm SDP}]
\right)$, where $I_{\rm AO}$ is the number of AO iterations.
Typically, the SDP subproblem $\mathbf P3$ dominates the
computational cost.

}

\section{Simulation Results}\label{sec5}

In this section, we present a comprehensive simulation study to evaluate the effectiveness of the proposed shared-aperture ISAC design under different system parameters and operating conditions. Unless explicitly stated otherwise, the simulation parameters are set according to Table~\ref{tab:sim_parameters}, including the antenna configuration, numbers of communication UAVs and sensing targets, number of sensing snapshots, BS transmit-power budget, downlink/uplink SINR thresholds, CRB threshold, and residual-SI level.

\begin{table}[t]
\centering
\caption{Simulation Parameter Settings}
\label{tab:sim_parameters}
\begin{tabular}{|l|c|}
\hline
Total number of BS antennas, \(M\) & 10 \\
\hline
Co-prime sensing pair, \((M_1,M_2)\) & \((3,4)\) \\
\hline
Number of communication UAVs, \(K_c\) & 2 \\
\hline
Number of sensing targets, \(K_s\) & 3 \\
\hline
Number of sensing snapshots, \(L\) & 256 \\
\hline
Maximum BS transmit power, \(P_{\mathrm{BS}}^{\max}\) & 40 dBm \\
\hline
SINR thresholds, 
\((\gamma_{\mathrm{DL}}^{\mathrm{th}},\gamma_{\mathrm{UL}}^{\mathrm{th}})\)
& \((10,10)\) dB \\
\hline
CRB threshold, \(\Gamma_{\mathrm{CRB}}\) 
& \(5\times10^{-6}\,\mathrm{rad}^{2}\) \\
\hline
Residual-SI strength, \(\rho_{\mathrm{SI}}\) & 0 dB \\
\hline
\end{tabular}
\end{table}

The downlink and uplink communication channels between the ground BS and the communication UAVs are modeled as i.i.d. Rayleigh fading vectors with unit average power, i.e.,
$\mathbf h_{c,k}^{\rm DL},\mathbf h_{c,k}^{\rm UL}
\sim\mathcal{CN}(\mathbf 0,\mathbf I)$,
where the large-scale path loss is absorbed into the channel normalization factor to isolate the impact of small-scale fading. For each simulated scenario, all numerical results are obtained by averaging over 500 independent channel realizations to ensure statistical reliability and reduce random fluctuations in the performance metrics.

\begin{figure}[t]
    \centering
    \includegraphics[width=\columnwidth]{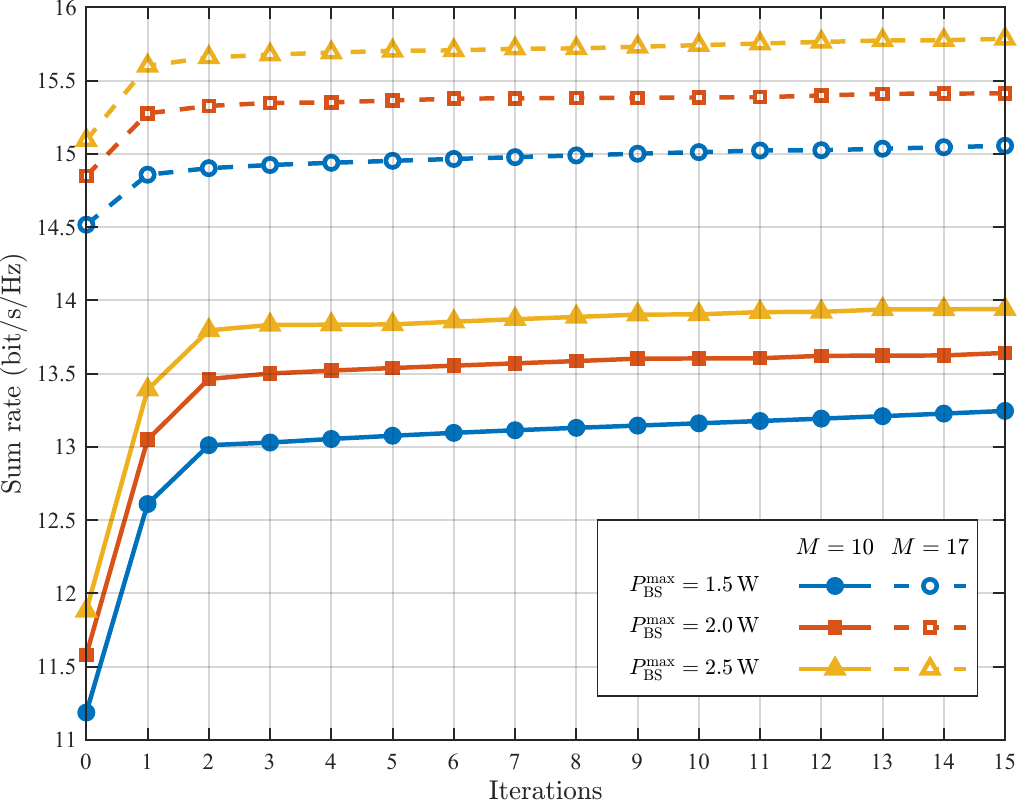}
    \caption{Convergence performance of Algorithm~1.}
    \label{fig:AO_convergence}
\end{figure}

To provide a fair and thorough performance comparison, we benchmark the proposed coprime-array-based ISAC beamforming scheme against several representative baseline array configurations. This comparison allows us to highlight the trade-offs among sensing accuracy, communication throughput, and hardware complexity, thereby demonstrating the advantages of the coprime array in the considered ISAC-UAV scenario.
\begin{itemize}
\item \emph{Baseline 1, Partitioned ULA subarrays} \cite{4350230}:
A ULA with the same total number of antennas $M$ as the proposed scheme is employed. Specifically, the array is divided into three contiguous subarrays assigned to sensing transmission, sensing reception, and TDD communication, respectively. The communication subarray operates as a transceiver, supporting downlink transmission and uplink reception in the corresponding TDD slots.

\item \emph{Baseline 2, Isotropic sensing transmission} \cite{10328645}:
This scheme adopts the same CPA sensing layout and communication-antenna 
positions as the proposed scheme. The sensing transmit covariance is 
restricted to an isotropic form
$\mathbf R_s^{\rm iso}=\frac{P_s}{M_1}\mathbf I_{M_1},$
where \(P_s\) denotes the sensing transmit power. The design is otherwise subject to the same sensing-accuracy, total BS transmit-power, downlink/uplink QoS, and residual-SI constraints as the proposed scheme. This baseline isolates the gain achieved by CRB-oriented sensing covariance optimization under the same CPA geometry.

\item \emph{Baseline 3, Minimum-Redundancy Array (MRA)} \cite{9352510}:
The sensing Tx and Rx positions are jointly determined through a combinatorial search on the common ULA grid to enlarge the contiguous aperture of the MIMO virtual array while reducing redundant virtual spacings. The same numbers of sensing Tx, sensing Rx, distinct sensing, and communication antennas as those of the proposed CPA architecture are maintained, and the remaining grid positions are assigned to the TDD communication transceiver subarray.

\end{itemize}

Fig.~\ref{fig:AO_convergence} illustrates the convergence behavior of Algorithm~1 under different maximum BS transmit power limitations and the number of antenna elements. It is observed that, across all parameter configurations, the system sum rate increases rapidly during the first few iterations. After that, the performance gain between successive iterations gradually diminishes, and the sum rate eventually goes to stable, which demonstrates that the proposed alternating optimization algorithm exhibits numerical convergence properties. As the maximum BS transmit power increases from $1.5~\mathrm{W}$ to $2.5~\mathrm{W}$, the converged sum rate improves notably. This is expected, since a larger transmission power provides greater optimization flexibility for the joint allocation of resources between the sensing waveform and the downlink communication beams. Furthermore, under the same transmit power level, the sum rate achieved with $M=17$ is consistently higher than that with $M=10$, indicating that increasing the array size enhances both communication beamforming capability and sensing performance. Overall, the majority of the performance gain is attained within the initial few iterations, confirming that Algorithm~1 is capable of converging rapidly to a stable solution under the considered transmit power and number of antenna elements.

\begin{figure}[t]
    \centering
    \includegraphics[width=\columnwidth]{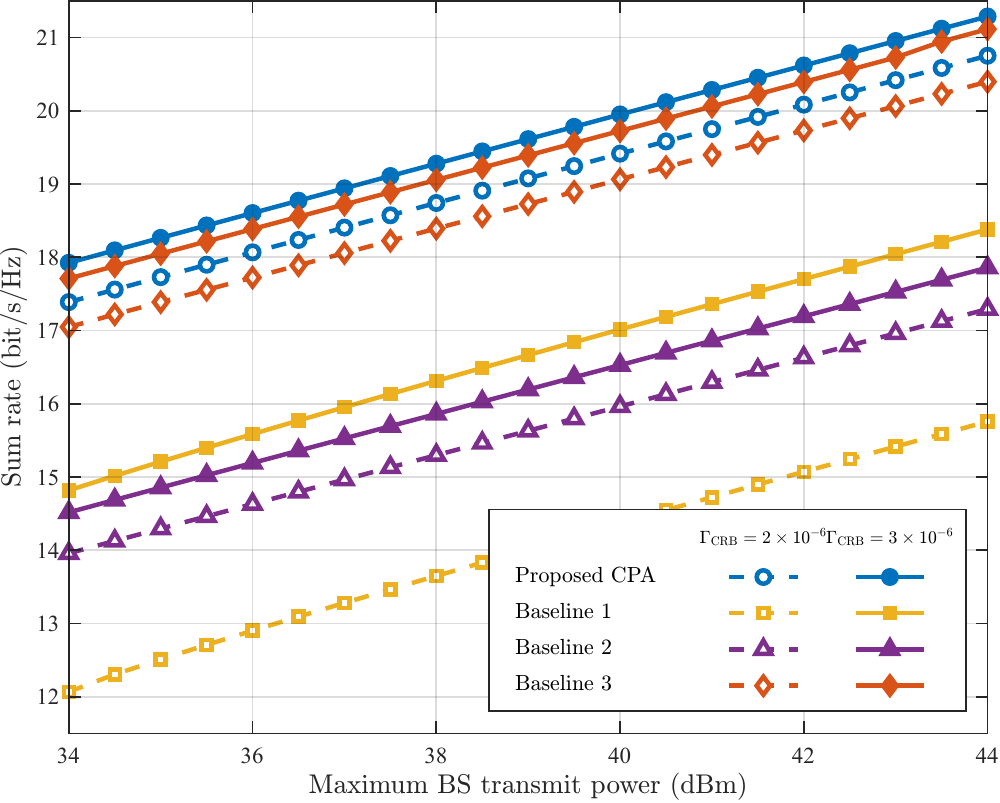}
    \caption{Sum-rate performance versus the maximum BS transmit power.}
    \label{fig:sum_rate_vs_power}
\end{figure}

\begin{figure}[t]
    \centering
    \includegraphics[width=\columnwidth]{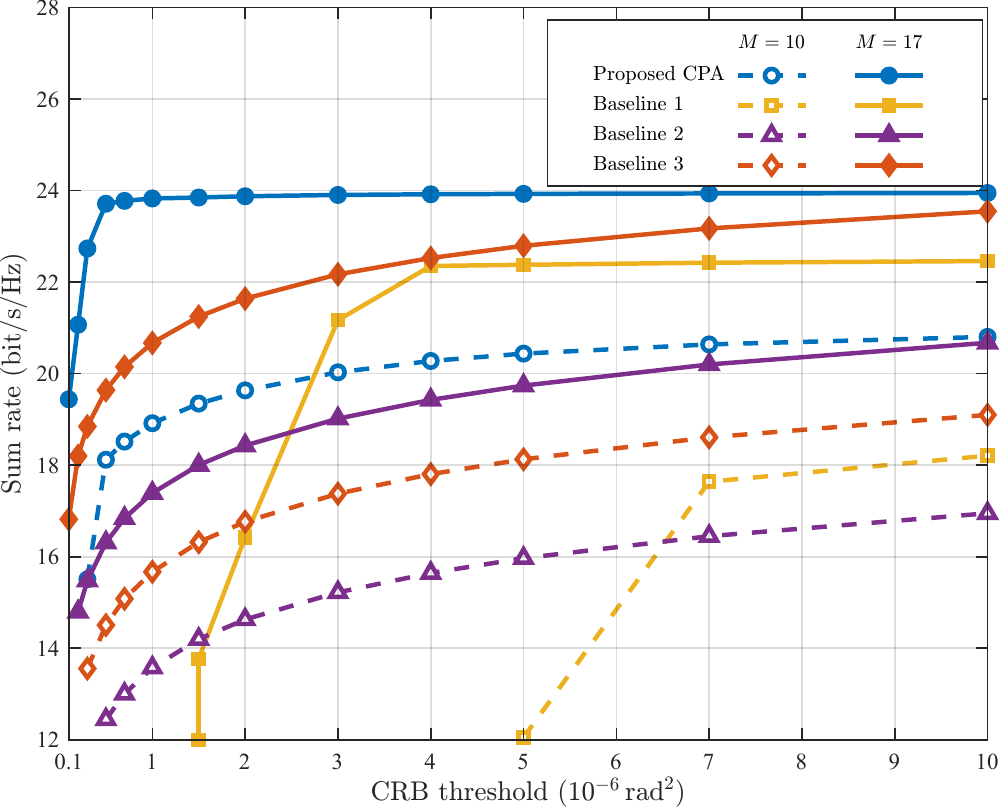}
    \caption{Sum-rate performance versus the CRB threshold.}
    \label{fig:sum_rate_vs_crb_threshold}
\end{figure}

\begin{figure}[t]
    \centering
    \includegraphics[width=\columnwidth]{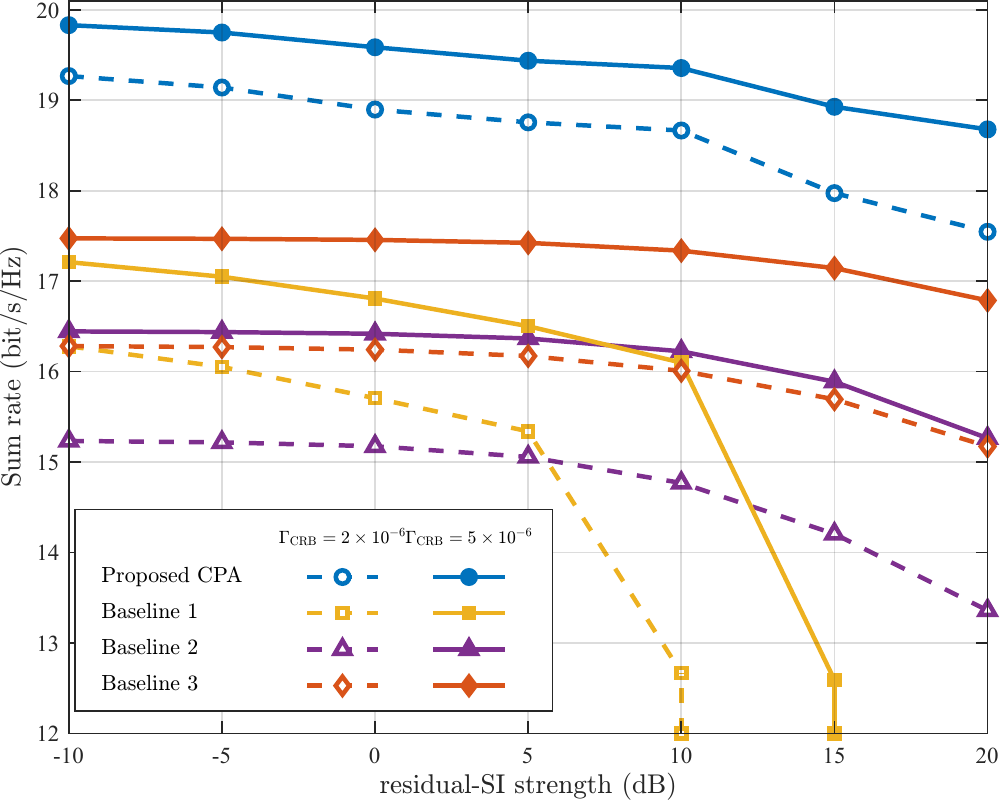}
    \caption{Sum-rate performance versus the residual SI power.}
    \label{fig:sum_rate_vs_communication_si}
\end{figure}

Fig.~\ref{fig:sum_rate_vs_power} presents the system sum rate as a function of the maximum BS transmit power under two different CRB thresholds. It is observed that, for all considered schemes, the sum rate increases monotonically with $P_{\mathrm{BS}}^{\max}$. This trend is attributed to the fact that a larger power budget expands the feasible resource-allocation space for the sensing waveform and the downlink communication beams, thereby enabling the system to achieve higher communication rates while still satisfying the constraints on communication QoS, residual self-interference (SI), and sensing accuracy. Under the same transmit-power level, the curves corresponding to $\Gamma_{\mathrm{CRB}}=3\times10^{-6}\,\mathrm{rad}^{2}$ consistently lie above those for $\Gamma_{\mathrm{CRB}}=2\times10^{-6}\,\mathrm{rad}^{2}$. This indicates that relaxing the sensing accuracy requirement reduces the resources required to satisfy the CRB constraint, thus enabling system allocating more power to data transmission.  The proposed co-prime array (CPA) scheme achieves the highest sum rate across the entire power range and consistently outperforms the minimum-redundancy array (MRA) baseline. This confirms that the CPA virtual array structure combined with joint resource allocation can effectively exploit the transmit power. In comparison with the partitioned ULA, the larger virtual aperture offered by the CPA lowers the resource cost of satisfying for the same CRB target. Moreover, the pronounced performance advantage of the proposed scheme over the CPA-Isotropic benchmark highlights that merely adopting the CPA architecture is insufficient to fully unlock its potential. The optimization of the sensing covariance matrix guaranteeing the sensing accuracy, communication interference, and residual SI enable to balance the tradeoff between sensing and communication..

Fig.~\ref{fig:sum_rate_vs_crb_threshold} illustrates the system sum rate as a function of the CRB threshold under different number of antenna elements. A smaller $\Gamma_{\mathrm{CRB}}$ corresponds to a more  strict requirement on angle-estimation accuracy. In this regime, the system must devote more communication resources to the sensing waveform, thereby limiting the design flexibility available for communication beamforming. As a result, the achievable sum rate decreases, and in some cases, certain schemes may even fail to satisfy all the constraints simultaneously. As $\Gamma_{\mathrm{CRB}}$ gradually increases, the sensing accuracy requirement is relaxed, allowing those schemes limited by the CRB constraint to allocate more resources to communication transmission. Consequently, the sum rate improves notably. {\color{black}As shown in Baseline 1, the problem becomes feasible only with larger CRB thresholds than those of the other schemes. Meanwhile, the proposed scheme can accommodate more strict CRB constraints.} It should be also noted that once the CRB constraint no longer constitutes the primary performance bottleneck, the curves tend to saturate.

It is worth noting that the proposed CPA scheme with $M=17$ maintains a nearly constant sum rate across the entire considered range of CRB thresholds. This observation indicates that the larger virtual aperture and richer angular observations provided by the CPA enable the system to meet stringent sensing requirements with little degradation in communication performance. Even with $M=10$, the proposed scheme remains feasible under relatively stringent CRB thresholds and achieves a higher sum rate than the baseline schemes with the same array size. In contrast, the partitioned ULA, particularly in the $M=10$ case, requires a considerably looser CRB threshold to become feasible. This simulation results show that its limited effective sensing aperture makes it difficult to simultaneously support high-accuracy sensing and maintain the required communication QoS.

\begin{figure}[t]
    \centering
    \includegraphics[width=\columnwidth]{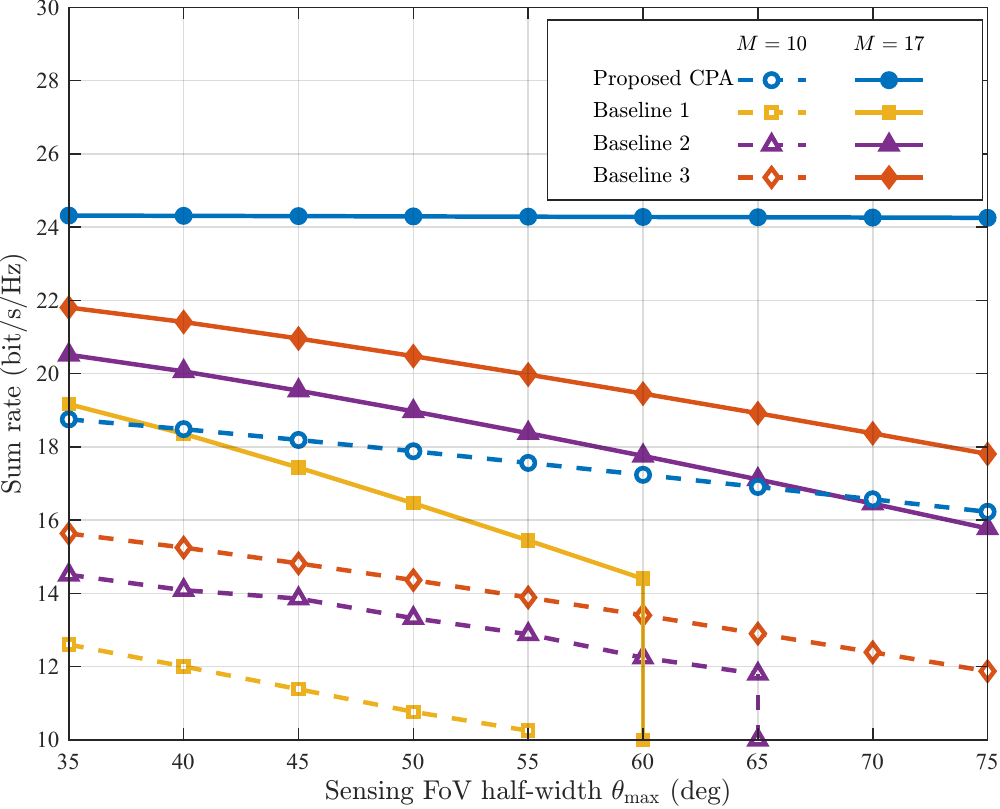}
    \caption{Sum-rate performance versus the sensing field-of-view half-width.}
    \label{fig:sum_rate_vs_sensing_fov}
\end{figure}

Fig.~\ref{fig:sum_rate_vs_communication_si} presents the system sum rate as a function of the residual SI level under different CRB thresholds. It is observed that, for all considered schemes, the sum rate decreases monotonically as the residual SI level increases. This degradation is attributed to the fact that stronger sensing-transmission leakage increases the interference level at both the sensing receive subarray and the communication transceiver antennas.  Under the same residual SI level, the looser CRB threshold of $\Gamma_{\mathrm{CRB}}=5\times10^{-6}\,\mathrm{rad}^{2}$ yields a higher sum rate. This is because the system is afforded greater flexibility in allocating sensing resources, thereby mitigating the communication performance loss imposed by stringent sensing accuracy requirements.
Meanwhile, the proposed CPA scheme consistently achieves the highest sum rate across the entire considered range of residual SI levels. Moreover, it remains feasible throughout this range and experiences only moderate performance degradation, demonstrating that the proposed design exhibits favorable robustness under imperfect SI cancellation. Compared with the CPA-Isotropic baseline, the performance advantage of the proposed scheme indicates that the jointly optimized sensing covariance matrix can effectively tailor the spatial distribution of the transmit energy according to the communication interference, sensing accuracy, and residual SI constraints, rather than relying solely on the geometric advantages of the CPA architecture. In contrast, the performance of the partitioned ULA deteriorates rapidly as the residual SI level increases, and infeasible points emerge at high residual SI levels. This highlights that its limited effective sensing aperture makes it difficult to simultaneously satisfy the communication QoS, sensing accuracy, and residual SI constraints.

\begin{figure}[t]
    \centering
    \includegraphics[width=\columnwidth]{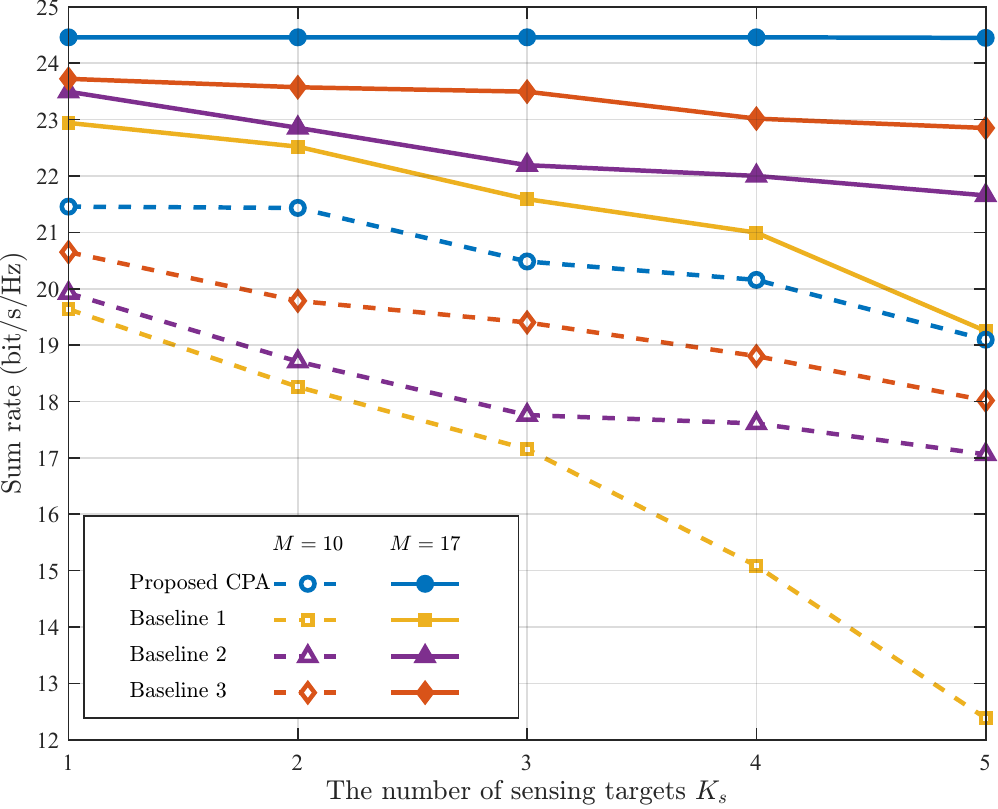}
    \caption{Sum-rate performance versus the number of sensing targets.}
    \label{fig:sum_rate_vs_sensing_targets}
\end{figure} 

Fig.~\ref{fig:sum_rate_vs_sensing_fov} illustrates the system sum rate as a function of the half-width of the sensing field of view, denoted by $\theta_{\max}$, under different array sizes. It is observed that, for most of the considered schemes, the sum rate gradually decreases as $\theta_{\max}$ increases. This is because a wider sensing field of view extends the target distribution over a larger off-axis angular range, thereby increasing the difficulty of satisfying the angle-estimation accuracy requirement. To satisfy the CRB constraint, the system must devote more transmit resources to the sensing task, leaving fewer resources available for communication transmission.
Notably, the proposed CPA scheme with $M=17$ maintains an approximately constant sum rate across the entire considered range of $\theta_{\max}$, indicating that the large virtual aperture of the CPA can support wide-FoV sensing with only a modest communication performance loss. Even with $M=10$, although the sum rate of the proposed scheme decreases somewhat as $\theta_{\max}$ increases, it remains feasible over the entire tested range and consistently outperforms the baseline schemes with the same array size.
Compared with other schemes, the partitioned ULA is the most sensitive to the expansion of the sensing field of view. Its sum rate decreases rapidly, and infeasible points emerge at large values of $\theta_{\max}$. This reveals that its limited effective sensing aperture makes it difficult to simultaneously maintain wide-FoV sensing accuracy and communication QoS. The performance advantage of the proposed scheme over the CPA-Isotropic baseline further indicates that, beyond the virtual-aperture gain provided by the CPA geometry, the optimization of the sensing covariance matrix---with explicit consideration of the target angles, communication interference, and residual SI---plays a critical role in sustaining wide-FoV sensing performance.

Fig.~\ref{fig:sum_rate_vs_sensing_targets} presents the system sum rate as a function of the number of sensing targets $K_s$ under different array sizes. It is observed that, for most of the considered schemes, the sum rate gradually decreases as $K_s$ increases. This is because a larger number of sensing targets increases the dimensionality of the angle parameters to be estimated and intensifies the coupling among different target parameters within the Fisher information matrix (FIM). Under a fixed CRB threshold, the system is required to simultaneously guarantee the angle-estimation accuracy for all targets, thereby devoting more transmit resources to the sensing task and leaving fewer resources available for communication transmission.
Notably, the proposed CPA scheme with $M=17$ maintains an approximately constant sum rate across the entire considered range of $K_s$, indicating that the large virtual aperture of the CPA can effectively alleviate the sensing-resource burden imposed by multi-target sensing. Even with $M=10$, although the sum rate of the proposed scheme decreases somewhat as $K_s$ increases, it consistently remains higher than those of the baseline schemes with the same array size.
Compared with other schemes, the partitioned ULA exhibits the most significant performance degradation and becomes infeasible at relatively small values of $K_s$. This reveals that its limited effective sensing aperture makes it difficult to simultaneously satisfy the multi-target CRB requirement, communication QoS, and residual SI constraints. The performance advantage of the proposed scheme over the CPA-Isotropic baseline further indicates that relying solely on the CPA geometry is insufficient to fully address the demands of multi-target sensing. Instead, it is essential to jointly optimize the sensing covariance matrix by taking into account the target directions, communication interference, and residual SI.

\section{Conclusion}

This paper proposed a shared-aperture ISAC architecture for UAV networks, where a sparse co-prime array is employed for full-duplex sensing while the remaining antenna positions within the same physical aperture are reused for communication. Through the CRB analysis, we showed that the enlarged virtual aperture provided by the CPA enables improved sensing performance compared with the partitioned ULA benchmark in both single-target and nondegenerate multi-target scenarios. Meanwhile, we revealed a space--time sampling tradeoff, where the sensing gain enabled by sparse spatial sampling generally requires additional temporal snapshots under the same physical aperture. Based on the proposed architecture, we further formulated a joint resource allocation problem to maximize the weighted communication sum rate under the sensing accuracy, transmit-power, communication QoS, and residual SI constraints, and developed an AO-based iterative algorithm for its solution. Simulation results showed that the proposed scheme consistently outperforms the considered baseline schemes and maintains favorable communication performance under increasingly demanding sensing and interference conditions. Moreover, the comparison with the isotropic CPA baseline demonstrated that the virtual-aperture gain alone is insufficient to fully exploit the potential of the CPA architecture, and that sensing covariance optimization plays an important role in improving the sensing--communication tradeoff. These results demonstrate the benefits of jointly exploiting sparse-array geometry and resource optimization for shared-aperture FD-ISAC systems.

\bibliographystyle{IEEEtran}
	%\bibliography{EE}
	%\begin{thebibliography}{99}
	%	
	%	%1
	%	\bibitem{ITU}
	%	 Technical Paper: Use Case & Requirements of Fibre-to-The-Room (FTTR), GSTP-FTTR ITU-T SG15 (2021).
	%	
	%\end{thebibliography}
	
	\bibliography{citepaper}                        %ref为.bib文件名

\clearpage
% ============================================================
% Appendices (merged from appendix.tex)
% ============================================================
\appendices
\section{Proof of Theorem 1}
\label{app:singal_target}
\setcounter{equation}{0}
\renewcommand{\theequation}{A.\arabic{equation}}
For the single-target case, the vectorized sensing mean is given by
\(
\tilde{\boldsymbol{\mu}}(\omega,\beta)
=
(\mathbf{X}^{T}\otimes \mathbf{I})\,\beta\,\boldsymbol{\phi}(\omega),
\)
where \(\beta=\beta_{\mathrm{R}}+j\beta_{\mathrm{I}}\). Using the mean-based FIM definition in (19),  together with the single-target parameter partition and Schur-complement formulation in (23)--(25), it remains to evaluate the FIM blocks associated with \(\omega\), \(\beta_{\mathrm{R}}\), and \(\beta_{\mathrm{I}}\). The corresponding partial derivatives are
\begin{align}
\frac{\partial \tilde{\boldsymbol{\mu}}}{\partial \omega}
&=
(\mathbf{X}^{T}\otimes \mathbf{I})\,\beta\,\dot{\boldsymbol{\phi}}(\omega),\\
\frac{\partial \tilde{\boldsymbol{\mu}}}{\partial \beta_{\mathrm{R}}}
&=
(\mathbf{X}^{T}\otimes \mathbf{I})\,\boldsymbol{\phi}(\omega),\\
\frac{\partial \tilde{\boldsymbol{\mu}}}{\partial \beta_{\mathrm{I}}}
&=
j(\mathbf{X}^{T}\otimes \mathbf{I})\,\boldsymbol{\phi}(\omega).
\end{align}
Substituting \(\partial \tilde{\boldsymbol{\mu}}/\partial \omega\) into (19) and adopting the snapshot approximation \(\mathbf{X}^{*}\mathbf{X}^{T}\approx L\mathbf{R}_{s}^{T}\), we obtain
\begin{equation}
F_{\omega\omega}
=
\frac{2L}{\sigma_r^2}
|\beta|^2
\dot{\boldsymbol{\phi}}^H(\omega)
(\mathbf{R}_s^T\otimes \mathbf{I})
\dot{\boldsymbol{\phi}}(\omega).
\label{eq:app_Fww_raw}
\end{equation}
Similarly,
\begin{align}
F_{\beta_{\mathrm{R}}\beta_{\mathrm{R}}}
&=
F_{\beta_{\mathrm{I}}\beta_{\mathrm{I}}}
=
\frac{2L}{\sigma_r^2}
\boldsymbol{\phi}^H(\omega)
(\mathbf{R}_s^T\otimes \mathbf{I})
\boldsymbol{\phi}(\omega),\\
F_{\omega\beta_{\mathrm{R}}}
&=
\frac{2L}{\sigma_r^2}
\Re\!\left\{
\beta^{*}
\dot{\boldsymbol{\phi}}^{H}(\omega)
(\mathbf{R}_{s}^{T}\otimes \mathbf{I})
\boldsymbol{\phi}(\omega)
\right\},\\
F_{\omega\beta_{\mathrm{I}}}
&=
-\frac{2L}{\sigma_r^2}
\Im\!\left\{
\beta^{*}
\dot{\boldsymbol{\phi}}^{H}(\omega)
(\mathbf{R}_{s}^{T}\otimes \mathbf{I})
\boldsymbol{\phi}(\omega)
\right\},
\end{align}
and \(F_{\beta_{\mathrm{R}}\beta_{\mathrm{I}}}=0\), since
\(
\boldsymbol{\phi}^{H}(\omega)(\mathbf{R}_s^{T}\otimes\mathbf{I})\boldsymbol{\phi}(\omega)
\)
is a Hermitian quadratic form and is therefore real-valued.

For compactness, let
\(
\mathbf{A}= \mathbf{R}_{s}^{T}\otimes \mathbf{I},
\)
\(
\mathbf{x}= \dot{\boldsymbol{\phi}}(\omega),
\)
\(
\mathbf{y}= \boldsymbol{\phi}(\omega),
\)
and define
\(
t= \mathbf{x}^{H}\mathbf{A}\mathbf{x},
\)
\(
q= \mathbf{x}^{H}\mathbf{A}\mathbf{y},
\)
\(
u= \mathbf{y}^{H}\mathbf{A}\mathbf{y}.
\)
Then, with the nuisance parameter vector
\(
\boldsymbol{\eta}=[\beta_{\mathrm{R}},\beta_{\mathrm{I}}]^{T},
\)
the nuisance block satisfies
\begin{equation}
\mathbf{F}_{\eta\eta}
=
\frac{2L}{\sigma_{r}^{2}}u\,\mathbf{I}_{2},
\end{equation}
and the Schur complement yields
\begin{equation}
F_{\omega\omega}
-
\mathbf{F}_{\omega\eta}\mathbf{F}_{\eta\eta}^{-1}\mathbf{F}_{\eta\omega}
=
\frac{2L}{\sigma_{r}^{2}}|\beta|^{2}
\left(
t-\frac{|q|^{2}}{u}
\right).
\label{eq:appA_schur_single}
\end{equation}

Moreover,
\begin{equation}
t-\frac{|q|^{2}}{u}
=
\mathbf{x}^{H}\mathbf{A}\mathbf{x}
-
\frac{|\mathbf{x}^{H}\mathbf{A}\mathbf{y}|^{2}}{\mathbf{y}^{H}\mathbf{A}\mathbf{y}}
=
\min_{c\in\mathbb{C}}
(\mathbf{x}-c\mathbf{y})^{H}\mathbf{A}(\mathbf{x}-c\mathbf{y}),
\end{equation}
where the minimizer is
\(c^\star =
\frac{\mathbf y^H\mathbf A\mathbf x}
{\mathbf y^H\mathbf A\mathbf y}\).
Under the setting in Theorem 1, the sensing transmit power of
each sensing transmit antenna is kept fixed as \(M_s\) grows.
Moreover, the sensing covariance \(\mathbf R_s\) is considered
full-rank and uniformly well-conditioned. Hence, its eigenvalues
are uniformly bounded above and below with respect to \(M_s\),
and
$\operatorname{tr}(\mathbf R_s)=\Theta(M_1),$
i.e., the total sensing transmit power grows linearly with the
number of sensing transmit antennas. Therefore, for any vector
\(\mathbf v\), we have
\[
\lambda_{\min}(\mathbf R_s)\|\mathbf v\|^2
\le
\mathbf v^H\mathbf A\mathbf v
\le
\lambda_{\max}(\mathbf R_s)\|\mathbf v\|^2.
\]

Using the standard identity
\(
\min_{c\in\mathbb{C}}\|\mathbf{x}-c\mathbf{y}\|^{2}
=
\|\mathbf{x}\|^{2}
-
\frac{|\mathbf{x}^{H}\mathbf{y}|^{2}}{\|\mathbf{y}\|^{2}},
\)
and restoring the original notation \(\mathbf{x}=\dot{\boldsymbol{\phi}}(\omega)\), \(\mathbf{y}=\boldsymbol{\phi}(\omega)\), we further obtain
\begin{equation}
\begin{aligned}
&\frac{2L}{\sigma_{r}^{2}}|\beta|^{2}\lambda_{\min}(\mathbf{R}_{s})
\Biggl(
\dot{\boldsymbol{\phi}}^{H}(\omega)\dot{\boldsymbol{\phi}}(\omega)
-
\frac{
\left|
\dot{\boldsymbol{\phi}}^{H}(\omega)\boldsymbol{\phi}(\omega)
\right|^{2}
}{
\boldsymbol{\phi}^{H}(\omega)\boldsymbol{\phi}(\omega)
}
\Biggr)
\\
&\le
F_{\omega\omega}-\mathbf{F}_{\omega\eta}\mathbf{F}_{\eta\eta}^{-1}\mathbf{F}_{\eta\omega}
\\
&\le
\frac{2L}{\sigma_{r}^{2}}|\beta|^{2}\lambda_{\max}(\mathbf{R}_{s})
\Biggl(
\dot{\boldsymbol{\phi}}^{H}(\omega)\dot{\boldsymbol{\phi}}(\omega)
-
\frac{
\left|
\dot{\boldsymbol{\phi}}^{H}(\omega)\boldsymbol{\phi}(\omega)
\right|^{2}
}{
\boldsymbol{\phi}^{H}(\omega)\boldsymbol{\phi}(\omega)
}
\Biggr).
\end{aligned}
\label{eq:appA_bound_proj}
\end{equation}

For the proposed CPA sensing architecture, the transmit and
receive steering vectors with respect to the spatial frequency
$\omega$ are given by
\begin{equation}
a_t(\omega)=
\big[1,e^{jM_2\omega},e^{j2M_2\omega},\ldots,e^{j(M_1-1)M_2\omega}\big]^T,
\label{eq:appA_at_cpa}
\end{equation}
\begin{equation}
a_r(\omega)=
\big[1,e^{jM_1\omega},e^{j2M_1\omega},\ldots,e^{j(M_2-1)M_1\omega}\big]^T,
\label{eq:appA_ar_cpa}
\end{equation}
and the corresponding virtual manifold is
\begin{equation}
\phi(\omega)=a_t^*(\omega)\otimes a_r^*(\omega).
\label{eq:appA_phi_cpa}
\end{equation}
Let
$D_t=\mathrm{diag}(0,M_2,2M_2,\ldots,(M_1-1)M_2),
D_r=\mathrm{diag}(0,M_1,2M_1,\ldots,(M_2-1)M_1),$
so that $\dot a_t(\omega)=jD_t a_t(\omega)$ and
$\dot a_r(\omega)=jD_r a_r(\omega)$. Hence,
\begin{equation}
\dot\phi(\omega)
=
\dot a_t^*(\omega)\otimes a_r^*(\omega)
+
a_t^*(\omega)\otimes \dot a_r^*(\omega).
\label{eq:appA_phi_dot_expand}
\end{equation}

To evaluate the projected term in \eqref{eq:appA_bound_proj}, it
suffices to compute the three inner products
$\phi^H(\omega)\phi(\omega)$,
$\dot\phi^H(\omega)\phi(\omega)$, and
$\dot\phi^H(\omega)\dot\phi(\omega)$.
Using the mixed-product property of the Kronecker product,
we first obtain
\begin{equation}
\phi^H(\omega)\phi(\omega)
=
\big(a_t^T(\omega)a_t^*(\omega)\big)
\big(a_r^T(\omega)a_r^*(\omega)\big)
=
M_1M_2.
\label{eq:appA_phi_phi_cpa}
\end{equation}

Next, by \eqref{eq:appA_phi_dot_expand},
\begin{align}
\dot\phi^H(\omega)\phi(\omega)
&=
\big(\dot a_t^H(\omega)a_t^*(\omega)\big)
\big(a_r^T(\omega)a_r^*(\omega)\big)
\nonumber\\
&\quad+
\big(a_t^T(\omega)a_t^*(\omega)\big)
\big(\dot a_r^H(\omega)a_r^*(\omega)\big).
\label{eq:appA_q_expand_cpa}
\end{align}
Using
$\sum_{m=0}^{M_1-1}m=\frac{M_1(M_1-1)}{2},
\sum_{n=0}^{M_2-1}n=\frac{M_2(M_2-1)}{2},$
we have
\begin{equation}
\dot\phi^H(\omega)\phi(\omega)
=
\frac{jM_1M_2}{2}\big(2M_1M_2-M_1-M_2\big),
\label{eq:appA_phiDot_phi_cpa}
\end{equation}
and hence
\begin{equation}
\frac{\left|\dot\phi^H(\omega)\phi(\omega)\right|^2}
{\phi^H(\omega)\phi(\omega)}
=
\frac{M_1M_2}{4}\big(2M_1M_2-M_1-M_2\big)^2.
\label{eq:appA_proj_second_term_cpa}
\end{equation}

Similarly, by expanding \eqref{eq:appA_phi_dot_expand} and using
$\sum_{m=0}^{M_1-1}m^2=\frac{M_1(M_1-1)(2M_1-1)}{6},
\sum_{n=0}^{M_2-1}n^2=\frac{M_2(M_2-1)(2M_2-1)}{6},$
straightforward algebra yields
% \begin{equation}
% \dot\phi^H(\omega)\dot\phi(\omega)
% =
% \frac{M_1M_2}{6}
% \Big[
% M_2^2(M_1-1)(2M_1-1)
% +
% M_1^2(M_2-1)(2M_2-1)
% +
% 3M_1M_2(M_1-1)(M_2-1)
% \Big].
% \label{eq:appA_phiDot_phiDot_cpa}
% \end{equation}
% 拆分后的长公式（双栏完美显示）
\begin{equation}
\begin{split}
\dot\phi^H(\omega)\dot\phi(\omega)
=
\frac{M_1M_2}{6}
\Big[
M_2^2(M_1-1)(2M_1-1)
\\
+
M_1^2(M_2-1)(2M_2-1) 
+
3M_1M_2(M_1-1)(M_2-1)
\Big].
\end{split}
\label{eq:appA_phiDot_phiDot_cpa}
\end{equation}

Substituting \eqref{eq:appA_phi_phi_cpa},
\eqref{eq:appA_proj_second_term_cpa}, and
\eqref{eq:appA_phiDot_phiDot_cpa} into the projected quantity in
\eqref{eq:appA_bound_proj}, we obtain
\begin{equation}
\small
\dot\phi^H(\omega)\dot\phi(\omega)
-
\frac{\left|\dot\phi^H(\omega)\phi(\omega)\right|^2}
{\phi^H(\omega)\phi(\omega)}
=
\frac{M_1M_2}{12}\big(2M_1^2M_2^2-M_1^2-M_2^2\big).
\label{eq:appA_proj_cpa_closed}
\end{equation}

Substituting $M_1=\rho(M_s+1)$ and $M_2=(1-\rho)(M_s+1)$ into
\eqref{eq:appA_proj_cpa_closed}, the leading-order term becomes
\begin{equation}
\dot\phi^H(\omega)\dot\phi(\omega)
-
\frac{\left|\dot\phi^H(\omega)\phi(\omega)\right|^2}
{\phi^H(\omega)\phi(\omega)}
\sim
\frac{1}{6}\rho^3(1-\rho)^3(M_s+1)^6,
\label{eq:appA_proj_cpa_asym}
\end{equation}
which implies
\begin{equation}
\dot\phi^H(\omega)\dot\phi(\omega)
-
\frac{\left|\dot\phi^H(\omega)\phi(\omega)\right|^2}
{\phi^H(\omega)\phi(\omega)}
=
\Theta(M_s^6).
\label{eq:appA_proj_cpa_theta}
\end{equation}
Together with \eqref{eq:appA_schur_single} and
\eqref{eq:appA_bound_proj}, this yields
\begin{equation}
\mathrm{CRB}_{\omega,\mathrm{CPA}}=\Theta(M_s^{-6}).
\label{eq:appA_crb_cpa_theta}
\end{equation}

% For the partitioned ULA sensing benchmark, let
% \begin{align}
% \bar{\mathbf{a}}_{t}(\omega)
% &=
% \big[
% 1,e^{j\omega},e^{j2\omega},\ldots,e^{j(M_{1}-1)\omega}
% \big]^{T},\\
% \bar{\mathbf{a}}_{r}(\omega)
% &=
% \big[
% 1,e^{j\omega},e^{j2\omega},\ldots,e^{j(M_{2}-1)\omega}
% \big]^{T},
% \end{align}
% and define
% \begin{equation}
% \bar{\boldsymbol{\phi}}(\omega)=\bar{\mathbf{a}}_{t}^{*}(\omega)\otimes \bar{\mathbf{a}}_{r}^{*}(\omega).
% \end{equation}
% Following the same derivation as in the CPA case, one obtains
% \begin{equation}
% \dot{\bar{\boldsymbol{\phi}}}^{H}(\omega)\dot{\bar{\boldsymbol{\phi}}}(\omega)
% -
% \frac{
% \left|
% \dot{\bar{\boldsymbol{\phi}}}^{H}(\omega)\bar{\boldsymbol{\phi}}(\omega)
% \right|^{2}
% }{
% \bar{\boldsymbol{\phi}}^{H}(\omega)\bar{\boldsymbol{\phi}}(\omega)
% }
% =
% \frac{M_{1}M_{2}}{12}\bigl(M_{1}^{2}+M_{2}^{2}-2\bigr).
% \end{equation}
% Hence, under
% \(
% M_{1}=\rho(M_{s}+1)
% \)
% and
% \(
% M_{2}=(1-\rho)(M_{s}+1)
% \),
% its leading-order term satisfies
% \begin{equation}
% \dot{\bar{\boldsymbol{\phi}}}^{H}(\omega)\dot{\bar{\boldsymbol{\phi}}}(\omega)
% -
% \frac{
% \left|
% \dot{\bar{\boldsymbol{\phi}}}^{H}(\omega)\bar{\boldsymbol{\phi}}(\omega)
% \right|^{2}
% }{
% \bar{\boldsymbol{\phi}}^{H}(\omega)\bar{\boldsymbol{\phi}}(\omega)
% }
% \sim
% \frac{1}{12}\rho(1-\rho)\bigl(\rho^{2}+(1-\rho)^{2}\bigr)(M_{s}+1)^{4},
% \end{equation}
% which, by the same argument, implies
% \begin{equation}
% \mathrm{CRB}_{\omega,\mathrm{ULA}}=\Theta(M_{s}^{-4}).
% \end{equation}

For the partitioned ULA sensing benchmark, let
\begin{equation}
\bar{a}_t(\omega)=
\big[1,e^{j\omega},e^{j2\omega},\ldots,e^{j(M_1-1)\omega}\big]^T,
\label{eq:appA_at_ula}
\end{equation}
\begin{equation}
\bar{a}_r(\omega)=
\big[1,e^{j\omega},e^{j2\omega},\ldots,e^{j(M_2-1)\omega}\big]^T,
\label{eq:appA_ar_ula}
\end{equation}
and define
\begin{equation}
\bar{\phi}(\omega)=\bar{a}_t^*(\omega)\otimes \bar{a}_r^*(\omega).
\label{eq:appA_phi_ula}
\end{equation}
Using the same asymptotic sequence  as in the CPA case, one can
similarly obtain the projected quantity associated with the ULA.
The intermediate expressions for
$\bar{\phi}^H(\omega)\bar{\phi}(\omega)$,
$\dot{\bar{\phi}}^H(\omega)\bar{\phi}(\omega)$, and
$\dot{\bar{\phi}}^H(\omega)\dot{\bar{\phi}}(\omega)$
follow analogously and are omitted for brevity. As a result,
\begin{equation}
\dot{\bar{\phi}}^H(\omega)\dot{\bar{\phi}}(\omega)
-
\frac{
\left|
\dot{\bar{\phi}}^H(\omega)\bar{\phi}(\omega)
\right|^2
}{
\bar{\phi}^H(\omega)\bar{\phi}(\omega)
}
=
\frac{M_1M_2}{12}\big(M_1^2+M_2^2-2\big).
\label{eq:appA_proj_ula_closed}
\end{equation}
Hence, under $M_1=\rho(M_s+1)$ and $M_2=(1-\rho)(M_s+1)$,
its leading-order term satisfies
\begin{equation}
\small
\dot{\bar{\phi}}^H(\omega)\dot{\bar{\phi}}(\omega)
-
\frac{
\left|
\dot{\bar{\phi}}^H(\omega)\bar{\phi}(\omega)
\right|^2
}{
\bar{\phi}^H(\omega)\bar{\phi}(\omega)
}
\sim
\frac{1}{12}\rho(1-\rho)\big(\rho^2+(1-\rho)^2\big)(M_s+1)^4,
\label{eq:appA_proj_ula_asym}
\end{equation}
which implies
\begin{equation}
\dot{\bar{\phi}}^H(\omega)\dot{\bar{\phi}}(\omega)
-
\frac{
\left|
\dot{\bar{\phi}}^H(\omega)\bar{\phi}(\omega)
\right|^2
}{
\bar{\phi}^H(\omega)\bar{\phi}(\omega)
}
=
\Theta(M_s^4).
\label{eq:appA_proj_ula_theta}
\end{equation}
Together with \eqref{eq:appA_schur_single} and
\eqref{eq:appA_bound_proj}, this yields
\begin{equation}
\mathrm{CRB}_{\omega,\mathrm{ULA}}=\Theta(M_s^{-4}).
\label{eq:appA_crb_ula_theta}
\end{equation}

\section{Proof of Theorem 2}
\label{app:multi_target}
\setcounter{equation}{0}
\renewcommand{\theequation}{B.\arabic{equation}}

For the multi-target case, define the parameter vector as
$\boldsymbol{\xi}
=
\begin{bmatrix}
\boldsymbol{\omega}^{T} &
\boldsymbol{\eta}^{T}
\end{bmatrix}^{T},
\boldsymbol{\eta}
=
\begin{bmatrix}
\Re\{\boldsymbol{\beta}\}^{T} &
\Im\{\boldsymbol{\beta}\}^{T}
\end{bmatrix}^{T},
\label{eq:appB_xi}$
where
\(
\boldsymbol{\omega}=[\omega_1,\ldots,\omega_{K_{\mathrm{s}}}]^{T}
\)
and
\(
\boldsymbol{\beta}=[\beta_1,\ldots,\beta_{K_{\mathrm{s}}}]^{T}
\).
Then, the full mean-based FIM can be partitioned as
\begin{equation}
\mathbf{F}
=
\begin{bmatrix}
\mathbf{F}_{\omega\omega} & \mathbf{F}_{\omega\eta}\\
\mathbf{F}_{\eta\omega} & \mathbf{F}_{\eta\eta}
\end{bmatrix}.
\label{eq:appB_F_block}
\end{equation}
By the block matrix inversion formula, the multi-target spatial-frequency CRB matrix is
\begin{equation}
\mathbf{CRB}_{\boldsymbol{\omega}}
=
\left(
\mathbf{F}_{\omega\omega}
-
\mathbf{F}_{\omega\eta}\mathbf{F}_{\eta\eta}^{-1}\mathbf{F}_{\eta\omega}
\right)^{-1}.
\label{eq:appB_CRB_block}
\end{equation}
Hence, it suffices to study the asymptotic scaling of
\begin{equation}
\widetilde{\mathbf{F}}_{\omega\omega}
=
\mathbf{F}_{\omega\omega}
-
\mathbf{F}_{\omega\eta}\mathbf{F}_{\eta\eta}^{-1}\mathbf{F}_{\eta\omega}.
\label{eq:appB_Feff_def}
\end{equation}

% Next, define
% $\mathbf{\Phi}
% =
% \big[
% \boldsymbol{\phi}(\omega_1),\ldots,\boldsymbol{\phi}(\omega_{K_{\mathrm{s}}})
% \big],
% \dot{\mathbf{\Phi}}
% =
% \big[
% \dot{\boldsymbol{\phi}}(\omega_1),\ldots,\dot{\boldsymbol{\phi}}(\omega_{K_{\mathrm{s}}})
% \big],$
% and let
% $\mathbf{A}= \mathbf{R}_{\mathrm{s}}^{T}\otimes \mathbf{I}.$

% Under the fixed sensing power budget, the eigenvalues of \(\mathbf{R}_{\mathrm{s}}\) remain uniformly bounded above and below. 
Next, define
\(\boldsymbol{\Phi}=
[\boldsymbol{\phi}(\omega_1),\ldots,
\boldsymbol{\phi}(\omega_{K_s})]\),
\(\dot{\boldsymbol{\Phi}}=
[\dot{\boldsymbol{\phi}}(\omega_1),\ldots,
\dot{\boldsymbol{\phi}}(\omega_{K_s})]\), and let
\(\mathbf A=\mathbf R_s^T\otimes\mathbf I\).
Under the setting in Theorems 1 and 2, the sensing transmit
power of each sensing transmit antenna is kept fixed as \(M_s\)
grows. Moreover, the sensing covariance \(\mathbf R_s\) is
considered full-rank and uniformly well-conditioned. Hence, its
eigenvalues are uniformly bounded above and below with respect
to \(M_s\), and
$\operatorname{tr}(\mathbf R_s)=\Theta(M_1).$
Therefore, inserting \(\mathbf{A}\) into the relevant quadratic forms changes only multiplicative constants and does not alter their asymptotic orders with respect to \(M_{\mathrm{s}}\). It thus suffices to characterize the corresponding inner products among \(\boldsymbol{\phi}(\omega_m)\) and \(\dot{\boldsymbol{\phi}}(\omega_m)\).

For any \(m\neq n\), Definition~1 implies that
\begin{equation}
\left|
\sin\!\left(\frac{M_1(\omega_m-\omega_n)}{2}\right)
\right|
\ge \delta,
\left|
\sin\!\left(\frac{M_2(\omega_m-\omega_n)}{2}\right)
\right|
\ge \delta.
\label{eq:appB_nondeg}
\end{equation}
This condition prevents the target-pair spatial-frequency difference from approaching the aliasing singularities induced by the sparse CPA geometry. 
Hence, the pairwise spatial-frequency difference
\(
\Delta\omega_{mn}\triangleq \omega_m-\omega_n
\)
stays away from the aliasing singularities induced by the two sparse spacings \(M_1\) and \(M_2\). As a result, the cross inner products can be reduced to oscillatory sums with respect to \(\Delta\omega_{mn}\), including the unweighted form \(\sum e^{j\ell\Delta\omega_{mn}}\), the first-order weighted form \(\sum \ell e^{j\ell\Delta\omega_{mn}}\), and the second-order weighted form \(\sum \ell^2 e^{j\ell\Delta\omega_{mn}}\). Under \eqref{eq:appB_nondeg}, the denominators of the corresponding geometric-series expressions remain bounded away from zero, and hence these sums do not add coherently. Therefore, for \(m\neq n\), one obtains
\begin{align}
\boldsymbol{\phi}^{H}(\omega_m)\mathbf{A}\boldsymbol{\phi}(\omega_n)
&= O(1), 
\label{eq:appB_cross_U}\\
\dot{\boldsymbol{\phi}}^{H}(\omega_m)\mathbf{A}\boldsymbol{\phi}(\omega_n)
&= O(M_{\mathrm{s}}^{2}), 
\label{eq:appB_cross_Q}\\
\dot{\boldsymbol{\phi}}^{H}(\omega_m)\mathbf{A}\dot{\boldsymbol{\phi}}(\omega_n)
&= O(M_{\mathrm{s}}^{4}), \qquad m\neq n.
\label{eq:appB_cross_T}
\end{align}
These bounds show that all inter-target cross terms are asymptotically weaker than their corresponding  diagonal counterparts.

We next consider the diagonal entries of \(\widetilde{\mathbf{F}}_{\omega\omega}\). Appendix~A has already established that, for the CPA sensing geometry,
\begin{equation}
\boldsymbol{\phi}^{H}(\omega_m)\mathbf{A}\boldsymbol{\phi}(\omega_m)=\Theta(M_{\mathrm{s}}^{2}),
\label{eq:appB_diag_U}
\end{equation}
and
\begin{equation}
\small
\dot{\boldsymbol{\phi}}^{H}(\omega_m)\mathbf{A}\dot{\boldsymbol{\phi}}(\omega_m)
-
\frac{
\left|
\dot{\boldsymbol{\phi}}^{H}(\omega_m)\mathbf{A}\boldsymbol{\phi}(\omega_m)
\right|^{2}
}{
\boldsymbol{\phi}^{H}(\omega_m)\mathbf{A}\boldsymbol{\phi}(\omega_m)
}
=
\Theta(M_{\mathrm{s}}^{6}).
\label{eq:appB_diag_proj}
\end{equation}
The second quantity is precisely the single-target projected term after eliminating the nuisance reflection coefficient. Since the cross-target contributions entering the \(m\)-th diagonal are lower-order by \eqref{eq:appB_cross_U}--\eqref{eq:appB_cross_T}, each diagonal entry of \(\widetilde{\mathbf{F}}_{\omega\omega}\) inherits the same leading-order scaling, namely,
\begin{equation}
\big[\widetilde{\mathbf{F}}_{\omega\omega}\big]_{m,m}
=
\Theta(M_{\mathrm{s}}^{6}),
\qquad
m=1,\ldots,K_{\mathrm{s}}.
\label{eq:appB_Feff_diag}
\end{equation}

We now turn to the off-diagonal entries. By construction, the blocks \(\mathbf{F}_{\omega\omega}\), \(\mathbf{F}_{\omega\eta}\), and \(\mathbf{F}_{\eta\eta}\) are all formed from the same cross inner products controlled above. Hence, the Schur correction term
\(
\mathbf{F}_{\omega\eta}\mathbf{F}_{\eta\eta}^{-1}\mathbf{F}_{\eta\omega}
\)
cannot increase the off-diagonal order beyond that already induced by the cross terms in \(\mathbf{F}_{\omega\omega}\). Consequently, after the exact Schur-complement reduction, the off-diagonal entries remain lower-order and satisfy
\begin{equation}
\big[\widetilde{\mathbf{F}}_{\omega\omega}\big]_{m,n}
=
O(M_{\mathrm{s}}^{4}),
\qquad
m\neq n.
\label{eq:appB_Feff_offdiag}
\end{equation}
Therefore, for fixed \(K_{\mathrm{s}}\), the matrix \(\widetilde{\mathbf{F}}_{\omega\omega}\) is dominated by its diagonal principal terms for sufficiently large \(M_{\mathrm{s}}\), while the off-diagonal entries contribute only lower-order perturbations. Hence, the inverse matrix is governed by the inverse order of the diagonal principal terms, and we obtain
\begin{equation}
\operatorname{tr}\!\left(\mathbf{CRB}_{\boldsymbol{\omega},\mathrm{CPA}}\right)
=
\operatorname{tr}\!\left(\widetilde{\mathbf{F}}_{\omega\omega}^{-1}\right)
=
O(M_{\mathrm{s}}^{-6}).
\label{eq:appB_trace_CPA}
\end{equation}
This proves the CPA part of Theorem~2.

For the partitioned ULA sensing benchmark, the same proof procedure applies. In this case, the corresponding target placement is assumed to stay away from the ULA degeneracy points, i.e., the pairwise spatial-frequency differences remain bounded away from the singular locations of the associated ULA oscillatory sums. Under this condition, the diagonal entries satisfy
\begin{equation}
\big[\widetilde{\mathbf{F}}_{\omega\omega}\big]_{m,m}
=
\Theta(M_{\mathrm{s}}^{4}),
\qquad
m=1,\ldots,K_{\mathrm{s}},
\label{eq:appB_Feff_diag_ULA}
\end{equation}
whereas the off-diagonal entries remain two orders weaker, namely,
\begin{equation}
\big[\widetilde{\mathbf{F}}_{\omega\omega}\big]_{m,n}
=
O(M_{\mathrm{s}}^{2}),
\qquad
m\neq n.
\label{eq:appB_Feff_offdiag_ULA}
\end{equation}
Therefore, for sufficiently large \(M_{\mathrm{s}}\), the inverse matrix is again determined by the inverse order of the diagonal principal terms, which yields
\begin{equation}
\operatorname{tr}\!\left(\mathbf{CRB}_{\boldsymbol{\omega},\mathrm{ULA}}\right)
=
O(M_{\mathrm{s}}^{-4}).
\label{eq:appB_trace_ULA}
\end{equation}
Combining \eqref{eq:appB_trace_CPA} and \eqref{eq:appB_trace_ULA} completes the proof.

	%\newpage

\end{document}